\documentclass[10pt,twocolumn]{article}

\usepackage[utf8]{inputenc}
\usepackage[T1]{fontenc}
\usepackage{lmodern}
\usepackage[margin=0.75in]{geometry}
\usepackage{amsmath,amssymb,amsthm}
\usepackage{graphicx}
\usepackage{booktabs}
\usepackage{longtable}
\usepackage{tabularx}
\usepackage{multirow}
\usepackage{hyperref}
\usepackage{cleveref}
\usepackage{enumitem}
\usepackage{xcolor}
\usepackage{listings}
\usepackage{float}
\usepackage{caption}
\usepackage{subcaption}
\usepackage{url}
\usepackage{xurl}  
\usepackage{microtype}
\usepackage{pifont}
\usepackage{wasysym}
\usepackage{fancyhdr}
\usepackage{titlesec}

\hypersetup{
  colorlinks=true,
  linkcolor=blue!70!black,
  citecolor=green!50!black,
  urlcolor=blue!70!black,
}

\newtheorem{theorem}{Theorem}
\newtheorem{definition}{Definition}
\newtheorem{property}{Property}
\newtheorem{observation}{Observation}

\titlespacing*{\section}{0pt}{1.5ex plus 0.5ex minus 0.2ex}{1ex plus 0.2ex}
\titlespacing*{\subsection}{0pt}{1.2ex plus 0.4ex minus 0.2ex}{0.8ex plus 0.2ex}

\title{\textbf{Governing Dynamic Capabilities: Cryptographic Binding and\\Reproducibility Verification for AI Agent Tool Use}}

\author{
  Ziling Zhou\\
  Genupixel Technology Pte.\ Ltd.\\
  \texttt{ziling@genupixel.com}
}

\date{}

\begin{document}
\maketitle
\thispagestyle{empty}

\begin{abstract}
AI agents now dynamically acquire tools, orchestrate sub-agents, and transact across organizational boundaries---yet no existing security layer can verify \emph{what} an agent is capable of, \emph{whether} it executed what it claims, or \emph{what actually happened} in a multi-agent interaction. We trace this governance vacuum to a missing architectural distinction: the \textbf{capability-context separation}. Inside a transformer's forward pass, tool definitions and user context are indistinguishable token sequences; at the orchestration layer, they have fundamentally different security semantics---tool definitions determine which real-world actions are \emph{possible} (and change infrequently), while runtime context determines which actions are \emph{chosen} (and changes per interaction). Existing frameworks conflate the two, enabling \emph{silent capability escalation}---agents acquiring tools without invalidating credentials---and leaving cross-agent interactions without verifiable provenance.

From this principle, we derive three \textbf{Agent Governance Requirements}: capability integrity (G1) governs the capability envelope, behavioral verifiability (G2) ensures agents executed their declared computational process, and interaction auditability (G3) provides tamper-evident records for the runtime context---together defining \emph{what} an agent ecosystem must enforce, independent of \emph{how}. We prove two structural results with prescriptive architectural consequences: the \textbf{Chain Verifiability Theorem}, establishing that behavioral verification is a chain property (one unverifiable interior agent breaks end-to-end verification for all downstream nodes), and the \textbf{Bounded Divergence Theorem}, which transforms replay-based verification into a probabilistic safety certificate ($\epsilon \leq 1 - \alpha^{1/n}$). We validate the framework with two crypto-agnostic instantiations---basic (Ed25519, SHA-256, hash chains; 97\,\textmu s verify) and enhanced (BBS+ selective disclosure, Groth16 DV-SNARK; 13.8\,ms verify)---both satisfying the same nine security properties. A multi-provider reproducibility study (9 models, 7 providers) reveals $5.8\times$ variance in inference determinism, directly connecting model characteristics to governance architecture. End-to-end evaluation over 5--20 agent pipelines with real LLM calls confirms $<$0.02\% governance overhead, real-time detection of 7 end-to-end attack scenarios (covering G1, G2, and G3) with zero false positives, and post-hoc forensic tracing of runtime behavioral attacks validated against the Nemotron-AIQ dataset.
\end{abstract}

\section{Introduction}

\subsection{Agent Governance: A Missing Foundation}

AI agents are no longer isolated tools. Through protocols like MCP~\cite{mcp_spec} and A2A~\cite{a2a_spec}, modern agents dynamically discover tools, invoke external APIs, and orchestrate other agents at runtime. As these agents assume real-world responsibilities---managing credentials, executing financial transactions, accessing sensitive data---a fundamental question emerges: \emph{what properties must an agent ecosystem satisfy for humans to trust it?}

This question has been raised from multiple directions: regulatory mandates for traceability and conformity assessment (EU AI Act, Articles 12, 43, 72), standards initiatives identifying security controls and governance oversight~\cite{nist_caisi,trism}, and research calling for secure interaction protocols~\cite{open_challenges}, transparency~\cite{cheng_safe}, formal safety properties~\cite{allegrini_formal}, and operational safety frameworks~\cite{ghosh_framework}. Yet \textbf{no existing framework translates these aspirations into enforceable cryptographic infrastructure}. This is a \textbf{fundamental architectural gap}: current agents acquire tools dynamically, execute arbitrary models, and interact without verifiable records.

We propose three \emph{Agent Governance Requirements}---the \textbf{minimum conditions} under which an agent ecosystem can be considered \emph{governed}:

\begin{definition}[Agent Governance Requirements]
\label{def:governance}
A \emph{governed agent ecosystem} satisfies three properties:

\textbf{G1: Capability Integrity.} Every agent's identity is cryptographically bound to its complete capability set---the model it executes and every tool it can invoke. The identity is unforgeable and non-repudiable. Any change to the capability set invalidates existing credentials and requires explicit re-authorization by a human or organizational principal.

\textbf{G2: Behavioral Verifiability (Inference Authenticity).} An agent's declared computational process can be independently verified---ensuring that the agent actually executed what its credentials claim. This requirement is \emph{implementation-agnostic}: it may be satisfied by hardware attestation, software-based replay, or zero-knowledge proofs.

\textbf{G3: Interaction Auditability.} All inter-agent and agent-human interactions produce tamper-evident records sufficient for forensic reconstruction, allowing humans to audit agent behavior, attribute responsibility, and provide evidence for regulatory compliance and system improvement.
\end{definition}

These requirements reflect the consensus of legal~\cite{hannecke,future_society}, regulatory~\cite{nist_caisi}, and technical~\cite{open_challenges,trism,cheng_safe,allegrini_formal,ghosh_framework} communities. What is novel is the claim that \emph{all three must hold simultaneously} and that \emph{none is achievable with existing infrastructure}. Table~\ref{tab:governance_gap} maps each requirement to the 12 attack scenarios in \S\ref{sec:threat}: no single requirement, nor any pair, suffices. Only G1$\wedge$G2$\wedge$G3 closes all 12 vectors.

\begin{table}[t]
\centering
\caption{Governance gap analysis: undefended attacks when one or two governance requirements are absent. ATK-9 (blame shifting) requires both G2 and G3 for full resolution. All 12 attacks are defined in \S\ref{sec:threat}.}
\label{tab:governance_gap}
\small
\begin{tabular}{@{}lll@{}}
\toprule
\textbf{Missing} & \textbf{Undefended Attacks} & \textbf{\#} \\
\midrule
\multicolumn{3}{@{}l}{\emph{Single requirement absent}} \\
$\neg$G1 & ATK-1,3,4,5,6,7 & 6 \\
$\neg$G2 & ATK-2,9,12 & 3 \\
$\neg$G3 & ATK-8,9,10,11 & 4 \\
\midrule
\multicolumn{3}{@{}l}{\emph{Two requirements absent}} \\
$\neg$G1,$\neg$G2 & ATK-1,2,3,4,5,6,7,9,12 & 9 \\
$\neg$G1,$\neg$G3 & ATK-1,3,4,5,6,7,8,9,10,11 & 10 \\
$\neg$G2,$\neg$G3 & ATK-2,8,9,10,11,12 & 6 \\
\midrule
G1$\wedge$G2$\wedge$G3 & \emph{none} & 0 \\
\bottomrule
\end{tabular}
\end{table}

A critical corollary: \textbf{even hardware TEEs~\cite{sgx_explained,arm_trustzone,nvidia_cc} do not eliminate the need for G1 and G3}. A TEE-attested agent that silently acquires new tools remains a security threat (G1 violation); a TEE-attested pipeline without tamper-evident records remains unauditable (G3 violation). Our framework provides the governance layer that TEE, by itself, cannot.

\subsection{Why the Current Landscape Falls Short}
\label{sec:intro_landscape}

The current agent security landscape maps into four layers (Figure~\ref{fig:four_layer}), none satisfying G1--G3:
\emph{Identity platforms}~\cite{keyfactor,spiffe_hashicorp,spiffe_solo} bind agents to keys but certificates remain valid after capability changes (\textbf{G1 violated}).
\emph{Authorization frameworks}~\cite{oidc_a,south_delegation,cerbos} evaluate permissions at grant time without detecting drift (\textbf{G1 violated}).
\emph{Runtime gateways}~\cite{gravitee,lasso} inspect tool calls statelessly (\textbf{G2, G3 violated}).
\emph{Governance platforms}~\cite{credo_ai} maintain registries without cryptographic enforcement (\textbf{G1 violated}).

\begin{figure}[t]
  \centering
  \includegraphics[width=\columnwidth]{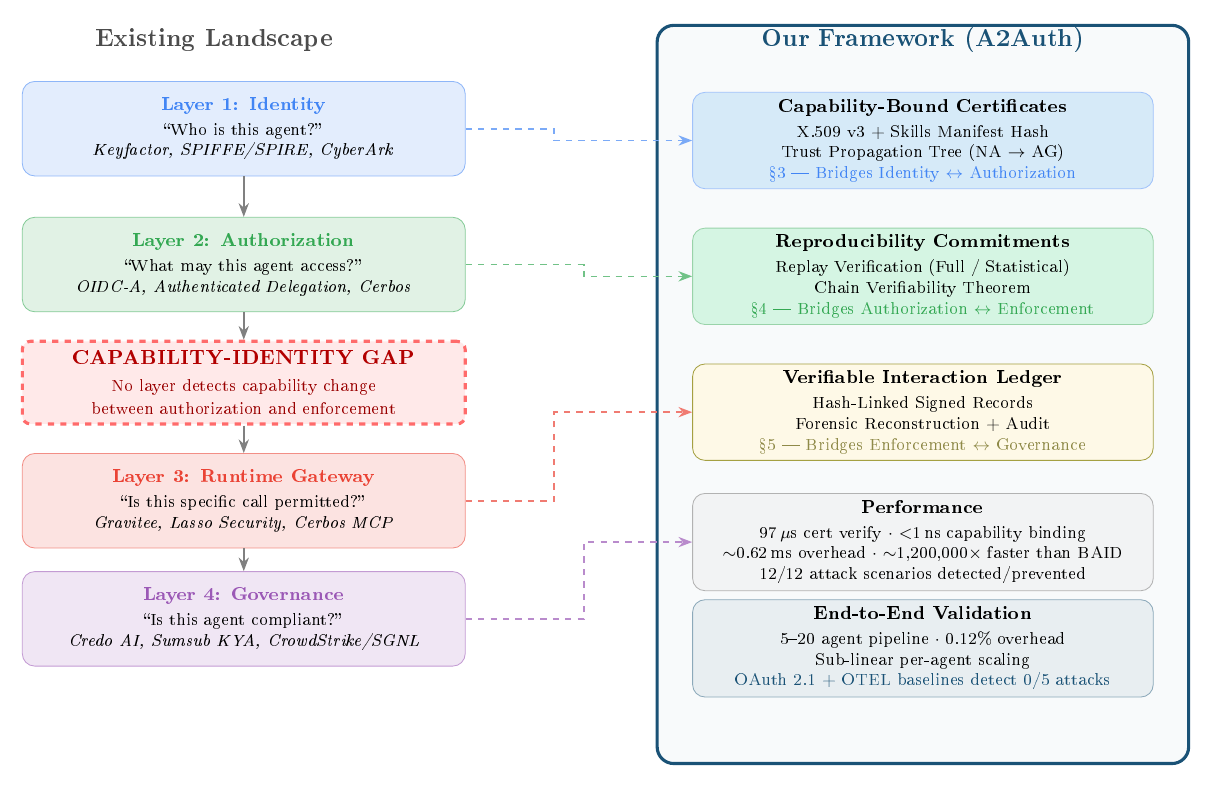}
  \caption{The four-layer agent security landscape and the capability-identity gap. Existing layers partially address individual requirements but none bridges them. Our framework provides the missing cryptographic governance layer satisfying G1--G3, validated end-to-end across 5--20 agent pipelines.}
  \label{fig:four_layer}
\end{figure}

The root cause is a failure to recognize a fundamental architectural distinction. Inside a transformer's forward pass, tool definitions and user context are indistinguishable token sequences; but at the orchestration layer, \textbf{tool definitions determine which real-world actions are possible}, while \textbf{runtime context determines which actions are chosen}. This yields the \emph{capability-context separation}: the \emph{capability envelope} (model identity, tool definitions, permission scopes) changes infrequently and requires re-authorization (G1); the \emph{runtime context} (conversation history, user instructions) changes per-interaction and requires verifiability (G2) and auditability (G3).

Existing frameworks conflate the two---a tool addition and a new user message are treated identically. On the capability side, this creates the \textbf{capability-identity gap}: no existing framework detects \emph{when an agent's capabilities have changed}, enabling \emph{silent capability escalation}. On the context side, it means cross-agent interactions lack verifiable provenance---no framework provides tamper-evident records linking \emph{which agent did what, with which inputs, producing which outputs}---the foundation for accountability in multi-agent workflows.

To our knowledge, the capability-context separation is the first architectural abstraction that explains \emph{why} agent security frameworks systematically fail to address capability drift---and it is not merely descriptive but \emph{prescriptive}: it dictates the governance boundaries any secure agent system must enforce. As a concrete illustration, BAID~\cite{baid} achieves strong code-level authentication via zkVM, binding agent identity to a program commitment $C_P = \mathrm{CommitProg}(P)$. Yet $C_P$ does not encompass dynamically acquired tools: an agent's program binary remains unchanged when new tools are discovered via MCP~\cite{mcp_spec} or plugin registries, leaving the \emph{entire capability envelope} outside BAID's security perimeter. This is not an implementation oversight but a direct consequence of designing without the capability-context separation---the framework conflates code identity with capability identity. Our G1 requirement, derived from this abstraction, prescribes that governance must commit to the \emph{complete capability set} (model + tools + permissions), not merely the executing code.

\subsection{Our Framework: Bridging the Gap}

The capability-context separation provides the foundation for three enforcement mechanisms that instantiate the governance contracts G1--G3:

\textbf{Capability-Bound Agent Certificates (\S\ref{sec:certificates}) $\boldsymbol{\to}$ G1.} We extend X.509 v3 with a \emph{skills manifest hash}---a SHA-256 commitment to the agent's complete tool configuration---embedded within a \emph{trust propagation tree} where human principals form the roots and trust constraints propagate monotonically downward.

\textbf{Inference Authenticity (\S\ref{sec:reproducibility}) $\boldsymbol{\to}$ G2.} Two paths: \emph{hardware} (TEE attestation~\cite{sgx_explained,arm_trustzone}) and \emph{software} (reproducibility-based replay leveraging LLM inference near-determinism~\cite{difr,eigenai}---a verifier re-executes with recorded inputs and flags divergence beyond the agent's declared threshold).

\textbf{Verifiable Interaction Ledger (\S\ref{sec:ledger}) $\boldsymbol{\to}$ G3.} Interaction records containing agent IDs, certificate hashes, content commitments, and reproducibility anchors are linked by a hash chain with per-record signatures, storing only cryptographic commitments to preserve privacy.

Together, these mechanisms provide a \emph{stateful, cryptographic governance layer} satisfying G1--G3.

\textbf{Crypto-agnostic design.} G1--G3 define \emph{what} to enforce, not \emph{how}. The framework is \textbf{crypto-agnostic}: each requirement specifies a functional contract with pluggable cryptographic primitives---analogous to TLS cipher suite negotiation. We present a \emph{basic} instantiation (Ed25519, SHA-256, hash chains; ${\sim}97$\,\textmu s verification) and an \emph{enhanced} instantiation (BBS+ selective disclosure for G1, DV-SNARK zero-knowledge verification for G2; 13.8\,ms). Both satisfy G1--G3 with identical security semantics, differing only in privacy properties and cost. This separation of governance logic from cryptographic mechanism demonstrates that agent governance is a \emph{composable architecture} where organizations select backends matching their threat model.

\subsection{Threat Model and Security Guarantees}

We consider a PPT adversary who can compromise up to $k$ agent nodes, create malicious agents, inject or modify tools, manipulate network traffic (Dolev-Yao model), and perform \emph{silent capability escalation}; it cannot compromise non-agent nodes or break standard cryptographic primitives. Under this model, we establish seven security properties and two structural theorems: \emph{Capability Binding} (Property~\ref{thm:capability_binding}), \emph{Trust Containment} (Property~\ref{thm:trust_containment}), \emph{Chain Integrity} (Property~\ref{thm:chain_integrity}), and \emph{Delegation Depth} (Property~\ref{thm:depth}) enforce G1; \emph{Reproducibility Soundness} (Property~\ref{thm:reproducibility}), \emph{Bounded Divergence} (Theorem~\ref{thm:bounded_divergence_proof}), and \emph{Chain Verifiability} (Theorem~\ref{thm:chain_verifiability_proof}) enforce G2; \emph{Credential Isolation} (Property~\ref{thm:isolation}) and \emph{Ledger Integrity} (Property~\ref{thm:ledger}) enforce G3.

\subsection{Contributions}

Our contribution is a new \emph{governance architecture}---identifying \emph{what} must be governed, \emph{where} governance boundaries belong, and \emph{why} a specific composition is both necessary and sufficient---not new cryptographic primitives. This follows the pattern of Certificate Transparency~\cite{ct_laurie}, which composed Merkle trees and gossip protocols (both well-known) into infrastructure that fundamentally changed web PKI operation. The insight was in identifying \emph{what to compose and where to deploy it}, not in new primitives; CT was published at IEEE S\&P precisely because the architectural contribution---public auditability of CA behavior---was independently valuable. We make an analogous claim: the capability-context separation, the three governance requirements, and the structural results that connect them are independently valuable regardless of the specific cryptographic instantiation. Concretely:
\textbf{1.\ A new vulnerability class and architectural abstraction.}
We identify \emph{silent capability escalation}---agents acquiring or modifying tools at runtime without triggering existing security mechanisms (\S\ref{sec:intro_landscape})---and trace it to the \textbf{capability-context separation}: tool definitions and user context are indistinguishable inside a transformer but have fundamentally different security semantics at the orchestration layer~\cite{saltzer84,hardy88}. To our knowledge, this is the first architectural abstraction that identifies \emph{where} governance boundaries must be placed in agent systems and \emph{why} existing frameworks---including those with strong cryptographic foundations~\cite{baid}---systematically miss capability drift.

\textbf{2.\ Structural results with architectural consequences.}
The \textbf{Chain Verifiability Theorem} (\S\ref{sec:chain_verifiability}) shows that a single unverifiable interior agent breaks end-to-end behavioral verification, prescribing high-determinism models at interior trust-tree positions.
The \textbf{Bounded Divergence Theorem} (\S\ref{sec:indistinguishability}) establishes that an adversary passing $n$ replay checks is $(n, \epsilon)$-indistinguishable from the declared model ($\epsilon \leq 1 - \alpha^{1/n}$), framing software G2 as a probabilistic approximation of TEE.
The \textbf{Agent Governance Trilemma} (\S\ref{sec:trilemma}) proves that capability, performance, and security cannot be simultaneously maximized for Turing-complete agents (via Rice's theorem).

\textbf{3.\ Empirical discovery connecting model determinism to chain security.}
A reproducibility study (9 models, 7 providers, 15{,}120 comparisons) reveals $5.8\times$ variance in inference determinism, which---combined with Chain Verifiability---determines which trust-tree positions a model qualifies for.

\textbf{4.\ Two crypto-agnostic instantiations.}
Basic (Ed25519, SHA-256; 97\,\textmu s) and enhanced (BBS+, DV-SNARK; 13.8\,ms) instantiations both satisfy nine security properties (\S\ref{sec:proofs}). End-to-end evaluation over 5--20 agent pipelines with real LLM calls confirms $<$0.02\% governance overhead and real-time detection of 7 E2E attack scenarios (covering G1, G2, and G3) with zero false positives (\S\ref{sec:e2e}). An additional 9 structural governance attacks are validated as automated unit tests, and 3 runtime behavioral attacks are validated via post-hoc forensic tracing against Nemotron-AIQ traces (\S\ref{sec:eval}).

\section{Threat Model and Motivating Attacks}
\label{sec:threat}

\subsection{System Model}

We model the trust infrastructure as a rooted directed tree $T = (V, E, r, \tau)$, where $V$ is the set of all entities (CAs, humans, organizations, AI agents), $E$ is the set of certificate issuance relationships, $r$ is the root CA, and $\tau: V \to \{\mathit{NA}, \mathit{AG}\}$ classifies each node as Non-Agent ($\mathit{NA}$) or Agent ($\mathit{AG}$).

The tree satisfies a \emph{structural type constraint}: $\tau(u) = \mathit{AG} \implies \forall (u,v) \in E: \tau(v) = \mathit{AG}$. Agents can delegate to sub-agents, but cannot confer trust upon humans or legal entities.

Each agent node $v$ holds a Capability-Bound Certificate $\mathit{cert}_v$ and maintains a runtime tool configuration $\mathcal{T}_v(t)$ that may change over time. Each credential $c$ has a risk tier $\mathit{tier}_c \in \{T0, T1, T2, T3\}$ where $T0$ (human-only) is most sensitive.

\subsection{Adversary Model}

We define a PPT adversary $\mathcal{A}$ with capabilities:
\textbf{C1: Network Control} (Dolev-Yao).
\textbf{C2: Agent Compromise} of up to $k$ agent nodes.
\textbf{C3: Rogue Agent Creation} with fresh key pairs.
\textbf{C4: Silent Capability Escalation}---modifying tool configurations without changing identity credentials.
\textbf{C5: Interaction Injection}---fabricating or modifying ledger records.

$\mathcal{A}$ \textbf{cannot}: compromise non-agent nodes, forge digital signatures (EUF-CMA), find hash collisions, or reverse hash functions.

\subsection{Motivating Attack Scenarios}

We enumerate twelve attacks organized into four categories. ATK-1 through ATK-9 address structural governance threats; ATK-10 through ATK-12 address runtime behavioral attacks empirically validated by the Nemotron-AIQ dataset~\cite{ghosh_framework}, which demonstrates 100\% success rates for prompt injection and tool-mediated attacks against unprotected agentic workflows.

\textbf{Capability Attacks.}
\emph{ATK-1: Silent Capability Escalation.} Agent $v$ certified with $\{$\texttt{code\_gen}, \texttt{unit\_test}$\}$ acquires \texttt{web\_browser} via a plugin marketplace. Under every existing framework, $v$'s credentials remain valid. Our framework detects the skills manifest hash change.

\emph{ATK-2: Model Substitution.} $\mathcal{A}$ replaces the certified model with a fine-tuned variant. Detected via G2: hardware TEE attestation (deterministic), software replay verification divergence, or DV-SNARK CharMatch failure.

\emph{ATK-3: Tool Trojanization.} $\mathcal{A}$ replaces tool implementation (same name/version, different code). Detected via code-level hashing in the skills manifest.

\textbf{Delegation Attacks.}
\emph{ATK-4: Phantom Delegation.} Compromised agent (T2) issues T1 certificate to sub-agent. Prevented by monotonic trust propagation.
\emph{ATK-5: Chain Forgery.} Fabricated certificate chain. Prevented by signature verification at every level.
\emph{ATK-6: Depth Overflow.} Deep delegation hierarchy. Prevented by strictly decreasing depth counters.
\emph{ATK-7: Credential Collusion.} Two or more compromised agents at lower tiers (e.g., T3$+$T3) attempt to combine their credentials to access resources restricted to a higher tier (T1). Prevented by independent per-agent credential evaluation---each agent's tier and permissions are verified in isolation, so colluding agents gain no additional permissions beyond their individual certificates (Property~\ref{thm:isolation}).

\textbf{Forensic Attacks.}
\emph{ATK-8: Evidence Tampering.} Log modification after incident. Prevented by hash-linked ledger with bilateral signing.
\emph{ATK-9: Blame Shifting.} Dispute in multi-agent chain. Resolved via interaction commitments and reproducibility replay.

\textbf{Runtime Behavioral Attacks.}
\emph{ATK-10: Injection-Driven Knowledge Poisoning.} Adversarial content injected via tool outputs (e.g., retrieval results) propagates through the agent's reasoning chain, corrupting final outputs. Nemotron-AIQ traces show this attack succeeds in 100\% of undefended trials across knowledge-poisoning and phishing-link templates. Our ledger's content commitments ($\eta_{\mathit{in}}, \eta_{\mathit{out}}$) enable post-hoc identification of the exact span where poisoned content entered the chain.

\emph{ATK-11: Cascading Action Hijack.} Injected instructions persist across multiple agent workflow stages (INFECT template), causing the agent to refuse legitimate tasks or produce harmful content at downstream nodes. Our hash-chain ledger records every inter-span transition, enabling forensic tracing of attack propagation paths that OTEL-only audit trails cannot reliably provide.

\emph{ATK-12: Output Exfiltration via Rendering.} Agent outputs are manipulated to embed phishing links (PHISHING\_LINK) or malicious markdown (MARKDOWN) that exfiltrate data when rendered by downstream consumers. Detected post-hoc via reproducibility verification: replaying the interaction with the committed parameters produces benign output, diverging from the tampered record.

\section{Capability-Bound Agent Certificates}
\label{sec:certificates}

\subsection{Certificate Structure}

\begin{definition}[Capability-Bound Certificate]
A Capability-Bound Certificate for node $v$ is:
$$\mathit{cert}_v = (\mathit{id}_v, \mathit{id}_{\mathit{par}(v)}, \mathit{pk}_v, \mu_v, \sigma_v, \kappa_v, \rho_v, t_s, t_e, \mathit{sig})$$
where $\mu_v = (\mathit{provider}, \mathit{model\_id}, \mathit{model\_ver})$ is the \textbf{model binding}, $\sigma_v = \{(\mathit{sid}_i, \mathit{ver}_i, h_i, P_i)\}_{i=1}^{n}$ is the \textbf{skills manifest}, $\kappa_v = (\mathit{max\_tier}, \mathit{max\_depth}, \mathit{allowed\_models}, \mathit{max\_rate})$ are \textbf{trust constraints}, and $\rho_v = (\mathit{level}, \mathit{config})$ is the \textbf{reproducibility commitment}.
\end{definition}

\subsection{Skills Manifest and Capability Binding}

\begin{definition}[Skills Manifest Hash]
$$H_v^{\mathit{cert}} = \mathrm{SHA\text{-}256}(\mathrm{canonical}(\{(\mathit{sid}_i, \mathit{ver}_i, h_i, P_i)\}_{i=1}^{n}))$$
\end{definition}

At runtime: $H_v(t) = \mathrm{SHA\text{-}256}(\mathrm{canonical}(S_v(t)))$.

\textbf{Capability Binding Invariant.} Access permitted only if $H_v(t) = H_v^{\mathit{cert}}$. Any capability change invalidates access. Crucially, runtime context changes (memories, preferences) do not affect $H_v(t)$---the practical consequence of capability-context separation.

For open-source tools, $h_i$ is the SHA-256 hash of source code. For closed-source tools, $h_i$ is the hash of the public configuration descriptor (name, version, API schema).

\subsection{Trust Propagation Tree}

\begin{figure}[t]
  \centering
  \includegraphics[width=\columnwidth]{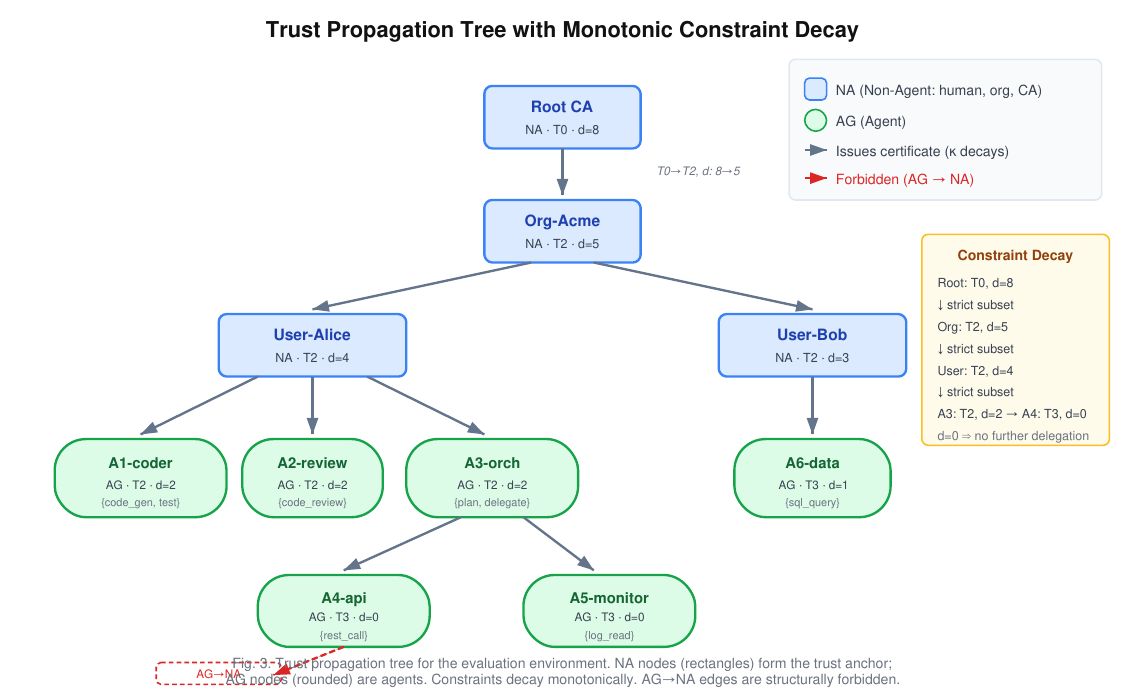}
  \caption{Trust propagation tree for the evaluation environment. NA nodes (rectangles) form the trust anchor; AG nodes (rounded) are agents. Constraints decay monotonically downward.}
  \label{fig:trust_tree}
\end{figure}

\begin{definition}[Trust Constraint Ordering]
$\kappa_a \leq_\kappa \kappa_b$ iff $\kappa_a.\mathit{max\_tier} \leq \kappa_b.\mathit{max\_tier}$, $\kappa_a.\mathit{max\_depth} < \kappa_b.\mathit{max\_depth}$ (strict), $\kappa_a.\mathit{allowed\_models} \subseteq \kappa_b.\mathit{allowed\_models}$, and $\kappa_a.\mathit{max\_rate} \leq \kappa_b.\mathit{max\_rate}$.
\end{definition}

\textbf{Trust Propagation Rule.} For every edge $(u, v) \in E$: $\kappa_v \leq_\kappa \kappa_u$.

\subsection{Certificate Verification Algorithm}

\begin{figure}[t]
  \centering
  \includegraphics[width=\columnwidth]{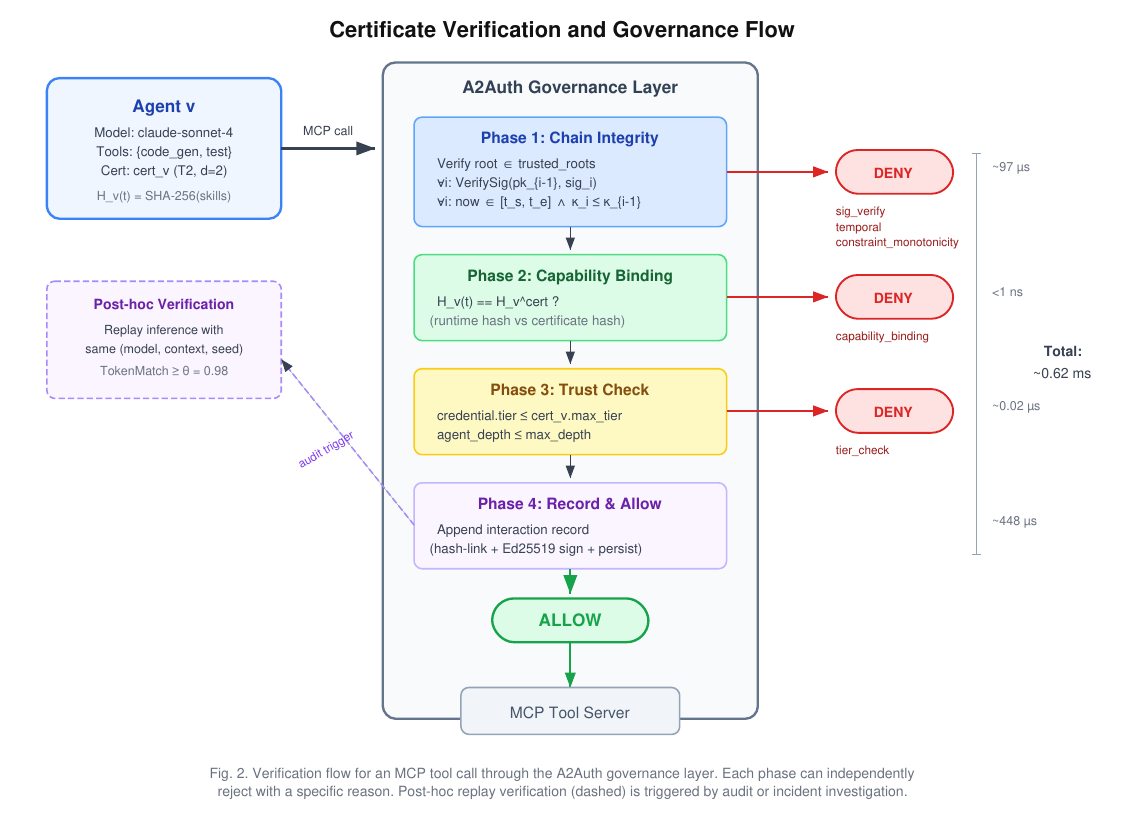}
  \caption{Verification flow for an MCP tool call through the A2Auth governance layer. Each phase can independently reject with a specific reason.}
  \label{fig:verification}
\end{figure}

When agent $v$ requests access to credential $c$ at time $t$, the verifier executes four phases (Figure~\ref{fig:verification}):

\begin{lstlisting}[language={},mathescape=true,basicstyle=\ttfamily\scriptsize]
function Verify(chain, credential, skills, roots):
  // Phase 1: Chain Integrity
  if chain[0].issuer $\notin$ roots: return DENY
  for i in 1..len(chain):
    if not VerifySig(chain[i-1].pk, chain[i].sig): return DENY
    if chain[i].$\kappa$ $\not\leq_\kappa$ chain[i-1].$\kappa$: return DENY
  // Phase 2: Capability Binding
  if SHA256(canonical(skills)) $\neq$ cert_v.$\sigma$_hash: return DENY
  // Phase 3: Trust Constraint Check
  if credential.tier > cert_v.$\kappa$.max_tier: return DENY
  // Phase 4: Revocation Check
  if IsRevoked(cert_v): return DENY
  return ALLOW
\end{lstlisting}

\subsection{Enhanced G1: Selective Disclosure via BBS+ Signatures}
\label{sec:bbs_plus}

The basic G1 instantiation (Ed25519 + SHA-256) requires \emph{full disclosure}: a verifier must see the entire skills manifest to recompute the hash and check capability binding. This suffices for single-organization deployments but is problematic in cross-organizational settings where revealing the full tool manifest leaks competitive intelligence (e.g., which proprietary tools an agent uses, its model version, or its rate limits).

\textbf{BBS+ signatures}~\cite{bbs_plus} enable a strictly stronger property: an agent can prove that its certificate is valid and that specific capability predicates hold, \emph{without revealing the full certificate contents}. BBS+ is a W3C Verifiable Credentials standard primitive with security reducing to the $q$-SDH assumption in bilinear groups~\cite{bbs_plus_security}.

\begin{definition}[BBS+ Capability Certificate]
\label{def:bbs_cert}
The issuer signs the certificate attributes as a vector of $L$ messages:
$$\mathbf{m} = (m_1, \ldots, m_L) = (\mathit{id}_v, \mathit{pk}_v, \mu_v, h_1, \ldots, h_n, \kappa_v, \rho_v, t_s, t_e)$$
producing a BBS+ signature $\Sigma = (A, e, s)$ over $\mathbf{m}$.
\end{definition}

\textbf{Selective disclosure.} To prove capability binding to a verifier, the agent derives a \emph{zero-knowledge proof of knowledge} $\pi$ that:
\begin{enumerate}[nosep]
  \item The agent possesses a valid BBS+ signature $\Sigma$ from a trusted issuer.
  \item The revealed subset $\mathbf{m}_R \subset \mathbf{m}$ (e.g., only the tool IDs, not model version or rate limits) matches the disclosed values.
  \item The hidden attributes $\mathbf{m}_H = \mathbf{m} \setminus \mathbf{m}_R$ satisfy declared predicates (e.g., $\mathit{max\_tier} \geq 2$, $\mathit{model} \in \{\texttt{gpt-4o}, \texttt{claude-4}\}$).
\end{enumerate}

The verifier learns only the revealed attributes and the predicate outcomes---not the hidden values. This enables scenarios such as: an agent proves it has \emph{some} valid tool manifest without revealing which tools; or proves its trust tier exceeds a threshold without revealing its exact tier or delegation depth.

\textbf{Capability binding under BBS+.} The capability binding invariant is preserved: at runtime, the agent recomputes $H_v(t) = \mathrm{SHA\text{-}256}(\mathrm{canonical}(S_v(t)))$ and includes a ZKP that $H_v(t)$ matches the committed hash in the signed message vector. Any tool change invalidates the proof.

\textbf{Cost-privacy tradeoff.} BBS+ selective disclosure proof generation costs ${\sim}14.7$\,ms and verification ${\sim}13.8$\,ms (benchmarked in \S\ref{sec:eval}), compared to ${\sim}97$\,\textmu s for basic Ed25519 verification---a ${\sim}142\times$ overhead. This is the concrete cost of privacy: the governance contract (G1) is identical; the privacy properties differ. Deployments choose between basic and enhanced G1 based on whether cross-organizational privacy justifies the latency cost. Crucially, these costs apply per-request only when selective disclosure is needed; for routine intra-organizational use, the basic instantiation suffices.

\textbf{Security.} BBS+ signature unforgeability reduces to the $q$-SDH assumption in Type-3 bilinear groups; the zero-knowledge property reduces to DDH. Both are standard assumptions with well-studied cryptanalytic history. The enhanced instantiation inherits all nine security properties of the basic instantiation, with the additional property of \emph{certificate privacy}: no PPT verifier can extract hidden attributes from a selective disclosure proof.

\subsection{Integration with A2A Agent Cards}

Agent Certificates are designed as non-invasive extensions of the A2A ecosystem via an optional \texttt{extensions} field in the Agent Card containing the certificate PEM, chain, and skills manifest hash. The enhanced instantiation adds a \texttt{disclosure\_mode} field (\texttt{full} or \texttt{selective}) and, for selective disclosure, the BBS+ proof $\pi$ replacing the raw certificate.

\section{Inference Authenticity and Reproducibility}
\label{sec:reproducibility}

\subsection{G2 as Inference Authenticity}

G2 requires \emph{inference authenticity}: verifying that an agent actually executed its declared model on the recorded inputs. This is an abstract governance requirement, independent of how it is achieved. Three implementation paths exist, each with distinct trust assumptions:

\textbf{Hardware attestation (TEE).} Trusted execution environments~\cite{sgx_explained,arm_trustzone} guarantee that the declared code was what actually executed. This provides the strongest G2 guarantee---deterministic, not probabilistic---but requires specialized hardware infrastructure (Intel SGX/TDX~\cite{sgx_explained}, ARM TrustZone~\cite{arm_trustzone}, NVIDIA H100 Confidential Computing~\cite{nvidia_cc}) that is unavailable in most current LLM API deployments. TEE satisfies G2 but \emph{does not satisfy G1 or G3}: a TEE-attested agent can still silently acquire tools (G1 violation), and TEE provides per-execution integrity without tamper-evident cross-agent audit trails (G3 violation). Our framework provides the governance layer that TEE alone cannot.

\textbf{Software replay verification.} When TEE is unavailable---the dominant case for cloud-hosted LLM APIs---G2 can be approximated through \emph{reproducibility-based replay}: re-execute the agent's declared model on recorded inputs and compare outputs. This provides a probabilistic software approximation of TEE, with guarantees formalized in Property~\ref{thm:reproducibility_informal} and Theorem~\ref{thm:bounded_divergence}. The approximation quality improves with the verification budget $n$.

\textbf{Zero-knowledge verification (DV-SNARK, \S\ref{sec:dv_snark}).} A privacy-preserving variant of replay verification: the prover demonstrates output consistency without revealing prompts or outputs. This adds privacy guarantees to the software path at the cost of computational overhead.

\textbf{Full-inference zkVM~\cite{baid}.} A SNARK circuit over the complete transformer forward pass provides the strongest software G2 guarantee---computational integrity without TEE hardware---by cryptographically binding the agent's execution to its registered code commitment $C_P = \mathrm{CommitProg}(P)$. This path is appropriate for high-assurance scenarios (e.g., financial settlement, regulated environments) where the proof generation overhead (${\sim}120$\,s per interaction) is acceptable and provider re-execution cooperation is unavailable. Unlike the software replay path, this path does not require model re-execution by the verifier; the proof is self-contained.

All four paths satisfy the same G2 governance contract. The framework is agnostic to which path is used; the choice is declared in the agent's certificate ($\rho_v$) and recorded in the ledger, making the G2 implementation transparent to auditors and downstream agents.

\subsection{Reproducibility as a Necessary Condition for Software G2}

For the software path, a key observation applies: \textbf{behavioral verifiability is impossible without reproducibility}. If an agent's outputs cannot be independently regenerated---even approximately---then no post-hoc mechanism can distinguish authorized behavior from unauthorized behavior. This is a fundamental constraint on \emph{any} software-based verification scheme.

This observation has a prescriptive implication for secure agent system design: \textbf{high-assurance deployments without TEE must be engineered for reproducible output.} Concretely, this means using deterministic sampling (temperature$=$0, fixed seeds), preferring structured outputs (code, JSON, SQL) over free-form prose where security-critical decisions are involved, and selecting inference providers that guarantee bitwise or near-bitwise reproducibility~\cite{eigenai}. Agents that cannot commit to any reproducibility level receive a trust tier downgrade (\S\ref{sec:repro_tier}), making the security cost of non-reproducibility explicit in the governance framework. More broadly, by making reproducibility a prerequisite for software G2 compliance, our framework transforms it from an implementation detail into a \emph{security-critical model performance metric}---with implications for how models are specified, evaluated, and deployed (\S\ref{sec:chain_verifiability}).

\subsection{Empirical Foundation}

LLM inference is \emph{nearly deterministic}. Karvonen et al.~\cite{difr} demonstrated that with fixed sampling parameters, regenerating output produces $>$98\% token-identical sequences. EigenAI~\cite{eigenai} achieved 100\% bit-exact deterministic inference on fixed GPU architectures with negligible performance loss. Our own validation (\S\ref{sec:eval}, Experiment~5) confirms that code-generating agents achieve \emph{perfect} determinism, while prose outputs exhibit paraphrase-level variation that remains well-separated from cross-model divergence.

\subsection{Reproducibility as a Trust Primitive}

\begin{definition}[Reproducibility Commitment]
$\rho_v = (\mathit{level}, \mathit{config})$ where $\mathit{level} \in \{\mathit{full}, \mathit{statistical}, \mathit{none}\}$.
\end{definition}

\textbf{Full} ($\mathit{level} = \mathit{full}$): $\forall (\mathit{context}, \mathit{seed}): O_v(\mathit{context}, \mathit{seed}) = O_v^{\mathit{replay}}(\mathit{context}, \mathit{seed})$.

\textbf{Statistical} ($\mathit{level} = \mathit{statistical}$):
$$\forall (\mathit{context}, \mathit{seed}): \mathrm{CharMatch}(O_v, O_v^{\mathit{replay}}) \geq \theta$$
The threshold $\theta$ is task-type dependent: our empirical validation (Experiment~5, \S\ref{sec:eval}) shows that code-generating agents achieve perfect determinism ($\theta = 0.98$), while prose-generating agents require $\theta = 0.10$ to account for paraphrase-level variation. Cross-provider mean CharMatch ($0.114$, Table~\ref{tab:exp5b}) remains well below even the relaxed threshold.

\textbf{None}: No reproducibility commitment.

\subsection{Reproducibility and Trust Tier Interaction}
\label{sec:repro_tier}

Reproducibility commitments interact with credential governance through a \emph{trust bonus}: agents with no reproducibility commitment are downgraded one tier (T1$\to$T2, T2$\to$T3). This makes the governance cost of non-reproducibility explicit: an agent that cannot be verified post-hoc is inherently less trustworthy, regardless of its other credentials.

\subsection{Replay-Based Verification Protocol}

\begin{lstlisting}[language={},mathescape=true,basicstyle=\ttfamily\scriptsize]
function ReplayVerify(cert, record, O_orig):
  $\rho$ := cert.$\rho$
  if $\rho$.level = none: return INCONCLUSIVE
  O_ref := ExecuteModel(cert.$\mu$, cert.$\sigma$,
           record.input, record.seed, $\rho$.config)
  if $\rho$.level = full:
    return O_ref == O_orig ? VERIFIED : VIOLATION
  else:  // statistical
    return CharMatch(O_ref, O_orig) $\geq$ $\theta$ ?
           VERIFIED : VIOLATION
\end{lstlisting}

\subsection{Security Analysis}

\begin{property}[Reproducibility Soundness, informal]
\label{thm:reproducibility_informal}
If compromised agent $v$ executes model $M' \neq M$ or skills $S' \neq S$, replay verification detects the deviation with probability $\geq 1 - \epsilon(\lambda)$ for negligible $\epsilon$.
\end{property}

By the DiFR token divergence result~\cite{difr}, the match rate between different models is $\delta \approx 0.06$ (character-level, validated in Experiment~5). Our empirical data (Table~\ref{tab:exp5c}) shows that the fraction of cross-model output pairs exceeding any reasonable threshold $\theta$ is at most $q \approx 0.02$ --- i.e., a substituted model's output coincidentally passes the CharMatch $\geq \theta$ check with probability at most $q$ per prompt. Over $n$ independent verification prompts, the probability of all checks passing despite model substitution is at most $q^n$, which is negligible for realistic $n$.

\textbf{Input integrity for replay verification.} Replay verification requires the recorded input to be authentic. A compromised agent cannot fabricate its own input because the interaction ledger uses \emph{bilateral signing}: the sender (an honest agent or non-agent) signs the input commitment $\eta_{\mathit{in}}$, which the compromised receiver cannot modify without invalidating the sender's signature. G2 therefore relies on \emph{sender-authenticated} input, not self-reported input---a unidirectional dependency on G3, not a circular one. The formal argument is given in Property~\ref{prop:g2_input_integrity} (\S\ref{sec:proofs}).

\subsection{Indistinguishability Under Sufficient Verification}
\label{sec:indistinguishability}

A natural concern is adversarial fine-tuning: could an adversary produce a model $M'$ that passes reproducibility checks on tested inputs but behaves maliciously on untested ones? We show that sufficient replay verification bounds the impact of any such attack.

\begin{definition}[$(n, \epsilon)$-Indistinguishability]
\label{def:indist}
Models $M$ and $M'$ are $(n, \epsilon)$-indistinguishable under skill set $\mathcal{S}$ and prompt distribution $\mathcal{D}$ if, after $M'$ passes replay verification on $n$ independent prompts $x_1, \ldots, x_n \sim \mathcal{D}$ across skills in $\mathcal{S}$:
$$\Pr_{x \sim \mathcal{D}}[\mathrm{CharMatch}(M(x), M'(x)) < \theta] \leq \epsilon$$
\end{definition}

\begin{theorem}[Bounded Divergence]
\label{thm:bounded_divergence}
If adversary-controlled model $M'$ passes replay verification (CharMatch $\geq \theta$) on $n$ independent prompts drawn from $\mathcal{D}$ across $|\mathcal{S}|$ skills, then $M$ and $M'$ are $(n, \epsilon)$-indistinguishable with:
$$\epsilon \leq 1 - \alpha^{1/n}$$
where $\alpha$ is the significance level. For small $\epsilon$, this simplifies to $\epsilon \lessapprox \ln(1/\alpha)/n$. Equivalently, for target failure-rate bound $\epsilon$, the required verification budget is $n \geq \ln(1/\alpha)/\epsilon$.
\end{theorem}

\textbf{Proof sketch.} Let $p = \Pr_{x \sim \mathcal{D}}[\mathrm{CharMatch}(M(x), M'(x)) < \theta]$ be the true failure rate under $M'$. Each verification trial $Z_i = \mathbf{1}[\mathrm{CharMatch} \geq \theta]$ is an independent Bernoulli$(1-p)$. Observing all $n$ tests pass ($\sum Z_i = n$), the likelihood is $(1-p)^n$. The exact one-sided Clopper-Pearson upper bound on $p$ at confidence $1-\alpha$ is $p \leq 1 - \alpha^{1/n}$. For practical computation: $(1-p)^n \leq \alpha$ implies $p \leq 1 - \alpha^{1/n} \approx \ln(1/\alpha)/n$ when $p$ is small. This is a distribution-free result: the bound depends only on the number of verification trials and the desired confidence level, not on the threshold $\theta$ itself. The role of $\theta$ is to define what constitutes a ``pass'' --- a higher $\theta$ increases sensitivity (more failures for a substituted model) but also increases the false-positive rate for the legitimate model. Full proof in \S\ref{sec:proofs}. \qed

\textbf{Practical implications.} Table~\ref{tab:verification_budget} shows the failure-rate bound $\epsilon$ (i.e., the maximum fraction of deployment prompts on which $M'$ can diverge undetected) as a function of the verification budget $n$ at 99\% confidence ($\alpha = 0.01$):

\begin{table}[h]
\centering\small
\begin{tabular}{rrl}
\toprule
$n$ & $\epsilon$ bound & Interpretation \\
\midrule
10 & 0.369 & Coarse screening \\
25 & 0.168 & Moderate assurance \\
50 & 0.089 & $<$10\% divergent prompts \\
100 & 0.045 & $<$5\% divergent prompts \\
200 & 0.023 & High assurance \\
500 & 0.009 & $<$1\% divergent prompts \\
\bottomrule
\end{tabular}
\caption{Failure-rate bound $\epsilon$ at 99\% confidence ($\alpha=0.01$) for various verification budgets $n$, computed as $\epsilon = 1 - 0.01^{1/n}$.}
\label{tab:verification_budget}
\end{table}

A verification suite of 50 prompts per skill --- operationally feasible as a periodic audit --- bounds the adversary's undetected divergence to at most 8.9\% of deployment prompts with 99\% confidence. Increasing to 200 prompts tightens this to 2.3\%.

\textbf{Threshold calibration.} The threshold $\theta$ must be calibrated per task type and model to balance detection sensitivity against the false-positive rate. Our empirical data (Experiment~5) shows that single-metric thresholds suffer from the bimodal output distribution: for code generation, CharMatch is either $>0.95$ (identical outputs) or $<0.5$ (paraphrase variants), leaving no single $\theta$ that simultaneously achieves low false positives and high detection. We therefore adopt a \textbf{multi-metric ensemble}: a prompt is flagged if \emph{any} of four metrics (CharMatch, Jaccard word overlap, TF-IDF cosine, character n-gram cosine) falls below its respective per-metric optimal threshold (determined by Youden's $J$ statistic, Table~\ref{tab:exp5e} in \S\ref{sec:eval}). This ensemble approach improves detection robustness across task types, as each metric captures different aspects of output similarity (surface-form identity vs.\ semantic overlap).

\textbf{Security interpretation.} This transforms reproducibility from a \emph{detection mechanism} into a \emph{safety certificate}: an adversary who fine-tunes $M$ into $M'$ that passes $n$ diverse replay checks has produced a model that diverges on at most an $\epsilon$-fraction of deployment prompts, with $\epsilon$ decreasing as $O(\ln(1/\alpha)/n)$. The adversary faces a dilemma: either $M'$ behaves like $M$ on most inputs (limiting the attack's impact to at most an $\epsilon$-fraction of requests), or $M'$ diverges more broadly and will be caught with probability increasing in $n$.

\subsection{Chain Verifiability}
\label{sec:chain_verifiability}

The preceding analysis establishes that G2 can verify an individual agent's computational process. A natural question arises: \emph{when agents form a multi-hop execution chain, is per-node verification sufficient for end-to-end behavioral verifiability?} We show that it is not---behavioral verifiability is a \textbf{chain property}, not a node property---and that this observation has direct implications for trust tree architecture (Figure~\ref{fig:chain_verifiability}).

\begin{figure}[t]
  \centering
  \includegraphics[width=\columnwidth]{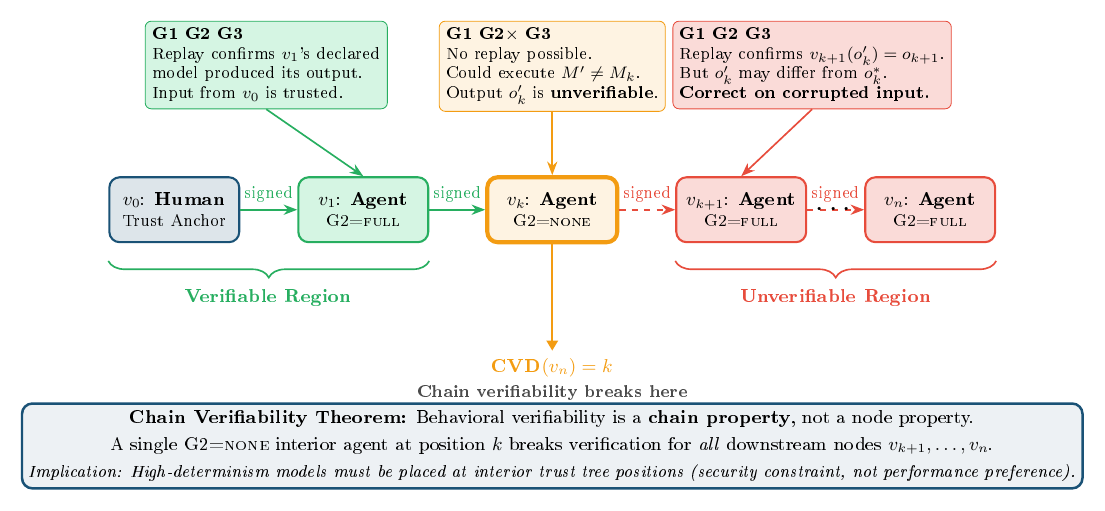}
  \caption{Chain Verifiability: a single G2$=$\textsc{none} interior agent at position $k$ breaks behavioral verification for all downstream nodes, even if they individually satisfy G2$=$\textsc{full}. The verifiable region extends only to CVD$(v_n) = k$.}
  \label{fig:chain_verifiability}
\end{figure}

\textbf{The gap that G2 uniquely fills.} G1+G3 together provide \emph{attribution integrity}---we know who sent what to whom, and that no message was altered---but signing proves attribution, not computational fidelity. A compromised agent holding a valid certificate (G1$\checkmark$) can execute an arbitrary model internally and sign the output (G3$\checkmark$); every cryptographic check passes. Only G2 detects this attack by verifying that the declared computational process produced the observed output, providing a probabilistic approximation of trusted execution ($\epsilon \leq 1 - \alpha^{1/n}$, Theorem~\ref{thm:bounded_divergence}) achievable purely in software.

\textbf{Why chain position matters.} Now consider a trust tree execution chain $v_0 \to v_1 \to \cdots \to v_n$, where $v_0$ is a human principal (trust anchor) and each $v_i$ ($i \geq 1$) is an agent node. Agent $v_j$ receives input from $v_{j-1}$, processes it using its declared model, and passes output to $v_{j+1}$.

Suppose $v_j$ is a compromised agent with G2$=$\textsc{none} (either because it uses a low-determinism model or because G2 verification is not enforced). Then:

\begin{enumerate}[nosep]
  \item $v_j$'s output cannot be verified as originating from its declared model.
  \item $v_j$'s output becomes $v_{j+1}$'s input. Even if $v_{j+1}$ satisfies G2$=$\textsc{full}, replay verification of $v_{j+1}$ only proves that $v_{j+1}$ faithfully processed the input it received---but that input itself is unverifiable.
  \item By induction, \emph{all} downstream nodes $v_{j+1}, \ldots, v_n$ operate on data whose provenance cannot be cryptographically verified back to the trust anchor.
\end{enumerate}

\begin{definition}[Chain Verifiability Depth]
\label{def:chain_depth}
For execution chain $v_0 \to v_1 \to \cdots \to v_n$ with $v_0$ a non-agent trust anchor, define:
$$\mathrm{CVD}(v_n) = \min\{j \geq 1 : v_j \text{ is an agent with } \rho_{v_j}.\mathit{level} = \mathit{none}\}$$
with $\mathrm{CVD}(v_n) = n$ if no such $v_j$ exists (fully verifiable chain).
\end{definition}

\begin{theorem}[Chain Verifiability]
\label{thm:chain_verifiability}
Let $v_0 \to v_1 \to \cdots \to v_n$ be an execution chain where $v_0$ is a non-agent trust anchor. If $\mathrm{CVD}(v_n) = k < n$, then for all $i > k$, the output of $v_i$ cannot be verified as originating from a legitimate computation chain rooted at $v_0$, even if $v_i$ individually satisfies $\rho_{v_i}.\mathit{level} = \mathit{full}$.
\end{theorem}

\textbf{Proof sketch.} By Definition~\ref{def:chain_depth}, $v_k$ has $\rho.\mathit{level} = \mathit{none}$, so no replay verification can confirm that $v_k$'s declared model produced its output $o_k$. Under the adversary model (\S\ref{sec:threat}), $v_k$ may have executed an arbitrary model $M' \neq M_k$ that produces output $o_k' \neq o_k^*$ (where $o_k^*$ is the output the declared model would have produced). Agent $v_{k+1}$ receives $o_k'$ as input. Even with $\rho_{v_{k+1}}.\mathit{level} = \mathit{full}$, replay verification confirms $v_{k+1}(o_k') = o_{k+1}$, i.e., that $v_{k+1}$ faithfully processed $o_k'$. But $o_k'$ may differ from $o_k^*$, so $o_{k+1}$ is the correct computation on \emph{incorrect} input. By induction on chain length, the same applies to all $v_i$ with $i > k$. Full proof in \S\ref{sec:proofs}. \qed

\textbf{Per-model reproducibility variance.} Our multi-provider experiment (Experiment~5c, \S\ref{sec:eval}) reveals that model determinism varies by over an order of magnitude: from CharMatch$=$0.987 (Haiku~4.5) to 0.169 (Kimi~K2.5). Combined with the Chain Verifiability Theorem, this transforms into \textbf{prescriptive architectural guidance}: the security of a multi-agent pipeline depends not just on whether G2 is enforced, but on \emph{where} high- and low-determinism models are placed in the trust tree.

\textbf{Reproducibility as a model performance metric.} The G1--G3 framework elevates reproducibility from an implementation detail to a \textbf{first-class model performance metric}: a model's determinism directly determines which governance level it qualifies for (\S\ref{sec:trilemma}), which trust-tree positions it may occupy, and what verification budget is needed (Theorem~\ref{thm:bounded_divergence}). We argue that providers should \textbf{self-declare reproducibility characteristics}---analogous to FIPS~140 security levels for cryptographic modules---including expected CharMatch ranges under temperature$=$0 and sensitivity to infrastructure changes. Our empirical characterization across 9~models and 7~providers (Experiment~5) provides the first multi-provider basis for such specifications.

\textbf{Architectural implications.} This theorem yields three deployment principles:

\emph{Principle 1: High-determinism models at interior nodes.} Models with strong reproducibility (Haiku: CharMatch$=$0.987, Sonnet: 0.823) occupy interior trust-tree positions; low-determinism models (Kimi: 0.169, DeepSeek: 0.225) are constrained to leaf positions---a \textbf{security architecture constraint}, not merely a performance recommendation.

\emph{Principle 2: Verification boundary checkpoints.} When a G2$=$\textsc{none} agent must appear at an interior position, the chain requires a verification boundary---human-in-the-loop, TEE attestation, or dual-execution---to re-establish chain verifiability for downstream nodes.

\emph{Principle 3: Explicit degradation in the ledger.} When chain verifiability breaks at node $v_k$, the ledger records a \texttt{PARTIAL\_VERIFIABILITY} marker; downstream agents inherit a trust tier penalty, making the reduced guarantee visible to auditors.

\textbf{Relationship to TEE.} The Chain Verifiability Theorem clarifies \emph{where} G2 is needed: at every interior node. In mixed deployments, different nodes may use different G2 implementations (TEE, software replay, or DV-SNARK), with chain verifiability depth determined by the weakest interior node. Regardless of the G2 implementation, each node still requires G1 and G3: TEE attests that \emph{some} model ran correctly, but only G1 ensures it was the authorized model, and only G3 provides tamper-evident records.

\subsection{Enhanced G2: Zero-Knowledge Behavioral Verification}
\label{sec:dv_snark}

The basic G2 instantiation requires \emph{provider cooperation}: the verifier must re-execute the agent's declared model on recorded inputs and compare outputs. This creates two limitations: (1)~the verifier must have API access to the same model version, and (2)~the original prompts and outputs are disclosed to the verifier, which may contain sensitive information.

\textbf{Designated-Verifier SNARKs} (DV-SNARKs) address both limitations by enabling the agent (or a delegated prover) to demonstrate output consistency \emph{in zero knowledge}: the verifier learns only that $\mathrm{CharMatch}(O, O^{\mathrm{replay}}) \geq \theta$, without seeing $O$, $O^{\mathrm{replay}}$, or the prompts.

\textbf{Key design insight.} Unlike BAID~\cite{baid}, which SNARKs over the entire transformer forward pass (${\sim}120$\,s per proof for a 7B model), our approach SNARKs only over the \emph{verification metric computation}---a fundamentally simpler circuit:

\begin{definition}[CharMatch Verification Circuit]
\label{def:charmatch_circuit}
$$\mathcal{C}_{\mathrm{verify}}(O, O^{\mathrm{replay}}, \theta) = \begin{cases} 1 & \text{if } \frac{|\{i : O[i] = O^{\mathrm{replay}}[i]\}|}{\max(|O|, |O^{\mathrm{replay}}|)} \geq \theta \\ 0 & \text{otherwise} \end{cases}$$
\end{definition}

For typical agent outputs ($|O| \leq 4{,}096$ characters), the circuit requires $O(|O|)$ equality comparisons and a single division---on the order of $10^4$ R1CS constraints (specifically: 19{,}238 for 1K outputs, ${\sim}73{,}000$ for 4K outputs), compared to ${\sim}10^{10}$ for a transformer circuit. This reduces proof generation from minutes to sub-seconds.

\textbf{Designated-Verifier property.} The proof is bound to a specific verifier's public key, preventing proof transfer. Only the designated auditor can verify the proof, which prevents a compromised agent from shopping for a lenient verifier. The DV property also prevents the proof from being used as transferable evidence---a deliberate design choice that limits the proof's scope to the specific audit context.

\textbf{Protocol.} Enhanced G2 verification proceeds as:
\begin{enumerate}[nosep]
  \item The agent (or a trusted replay service) executes the model on the recorded input to produce $O^{\mathrm{replay}}$.
  \item The prover computes $\pi = \mathrm{DV\text{-}SNARK}.\mathrm{Prove}(\mathcal{C}_{\mathrm{verify}}, (O, O^{\mathrm{replay}}), \theta, \mathit{pk}_{\mathrm{verifier}})$.
  \item The verifier checks $\mathrm{DV\text{-}SNARK}.\mathrm{Verify}(\pi, \theta, \mathit{sk}_{\mathrm{verifier}}) \stackrel{?}{=} 1$.
\end{enumerate}

\textbf{What DV-SNARK does and does not prove.} The SNARK proves that the prover possesses two strings whose CharMatch exceeds $\theta$; it does \emph{not} prove that $O^{\mathrm{replay}}$ was produced by the declared model (that trust is delegated to the replay service). The SNARK's contribution is \emph{privacy}; for full computational integrity, TEE or BAID-style full-inference SNARKs remain necessary.

\textbf{Cost-privacy tradeoff.} DV-SNARK proof generation costs ${\sim}174$\,ms (1K outputs) to ${\sim}750$\,ms (4K outputs); Groth16 verification is constant at ${\sim}1.3$\,ms (\S\ref{sec:eval}). Proof generation is ${\sim}1{,}740\times$ slower than basic CharMatch but is offline and asynchronous; it is ${\sim}690\times$ faster than BAID's full-inference zkVM (${\sim}120$\,s). The governance contract (G2) is identical; the privacy properties differ.

\textbf{Security.} Knowledge soundness ensures that a prover cannot produce a valid proof without possessing strings satisfying the CharMatch predicate. The DV property ensures simulation-privacy: no entity other than the designated verifier can distinguish a real proof from a simulated one.

\section{Verifiable Interaction Ledger}
\label{sec:ledger}

\subsection{Motivation}

Capability binding and reproducibility verify individual agents. Real-world incidents involve \emph{chains} of interactions. We need an immutable, tamper-evident record where no single party---including the ledger operator---can undetectably alter history.

\textbf{Why OpenTelemetry traces are insufficient.} The industry standard for agentic observability is OpenTelemetry (OTEL)~\cite{otel_spec}, whose span trees record execution traces with \texttt{span\_id}/\texttt{parent\_id} linkage, timestamped \texttt{input.value}/\texttt{output.value} payloads, and span-kind annotations (\texttt{CHAIN}, \texttt{LLM}, \texttt{TOOL}). Analysis of the Nemotron-AIQ Agentic Safety Dataset~\cite{ghosh_framework}---10,796 OTEL traces from NVIDIA's AI-Q research assistant under 22 distinct attack templates---reveals three critical gaps: (1)~\emph{No integrity protection}: span linkage is purely referential; any party with storage access can modify, delete, or reorder spans undetectably. (2)~\emph{No identity binding}: spans record function names but not cryptographically verified agent identities or capability certificates. (3)~\emph{No privacy preservation}: OTEL stores full plaintext input/output payloads, creating data exposure risk.
Our interaction ledger addresses all three gaps through hash-chain linking ($h_{\mathit{prev}}$), bilateral digital signing ($\mathit{sig}_i$; Ed25519 in the basic instantiation), agent certificate binding ($\mathit{ch}_{\mathit{send}}, \mathit{ch}_{\mathit{recv}}$), and commitment-only storage ($\eta_{\mathit{in}}, \eta_{\mathit{out}}$).

\subsection{Ledger Structure}

\begin{definition}[Interaction Record]
$R_i = (\mathit{seq}_i, t_i, a_{\mathit{send}}, a_{\mathit{recv}}, \mathit{ch}_{\mathit{send}}, \mathit{ch}_{\mathit{recv}}, \eta_{\mathit{in}}, \eta_{\mathit{out}}, \pi_i, h_{\mathit{prev}}, \mathit{sig}_i)$
where $\eta_{\mathit{in}} = \mathrm{SHA\text{-}256}(\mathit{input})$, $\eta_{\mathit{out}} = \mathrm{SHA\text{-}256}(\mathit{output})$, $\pi_i = (\mathit{seed}, \mathit{model\_ver}, \mathit{skills\_hash})$, and $h_{\mathit{prev}} = \mathrm{SHA\text{-}256}(R_{i-1})$.
\end{definition}

\begin{definition}[Verifiable Interaction Ledger]
$\mathcal{L} = (R_1, \ldots, R_n)$ satisfying: (1) sequential consistency, (2) hash chain integrity, and (3) signature validity at every record.
\end{definition}

\subsection{Privacy Preservation}

The ledger stores only \emph{commitments} (hashes), not content. During normal operation, only interaction metadata is revealed. During forensic investigation, parties produce original content verified against stored commitments.

\subsection{Forensic Reconstruction Protocol}

The protocol proceeds in five steps: (1) identify the relevant interaction chain, (2) verify ledger integrity (signatures + hash chain), (3) request content disclosure and verify against commitments, (4) verify agent certificates at interaction time, and (5) execute replay verification for agents with reproducibility commitments.

\subsection{Chain Auditability}
\label{sec:chain_auditability}

Analogous to the Chain Verifiability Theorem for G2 (\S\ref{sec:chain_verifiability}), interaction auditability is a \textbf{chain property}: a single missing interaction record along an execution path breaks forensic reconstruction for all downstream nodes on that path.

\begin{definition}[Chain Auditability Depth]
\label{def:cad}
For execution chain $v_0 \to v_1 \to \cdots \to v_n$ with $v_0$ a non-agent trust anchor, define:
\begin{multline*}
\mathrm{CAD}(v_n) = \min\{j \geq 1 : \\
\text{record for } v_{j-1} \to v_j \text{ absent from } \mathcal{L}\}
\end{multline*}
with $\mathrm{CAD}(v_n) = n$ if all records are present (fully auditable chain).
\end{definition}

\begin{property}[Chain Auditability]
\label{prop:chain_auditability}
If $\mathrm{CAD}(v_n) = k < n$, then for all $i \geq k$:
\begin{enumerate}[nosep]
  \item \emph{Provenance gap}: the auditor cannot trace $v_i$'s output back to the trust anchor $v_0$, because the input $v_k$ received from $v_{k-1}$ is unrecorded.
  \item \emph{G2 cascade}: replay verification for $v_k$ is impossible---the sender-authenticated input $\eta_{\mathit{in}}$ required by Property~\ref{prop:g2_input_integrity} is unavailable. By induction, downstream G2 verification inherits the same gap even if $\rho_{v_i}.\mathit{level} = \mathit{full}$.
\end{enumerate}
\end{property}

This creates a symmetric pair of chain properties: G2's Chain Verifiability (Theorem~\ref{thm:chain_verifiability}) governs \emph{behavioral verification depth}, while G3's Chain Auditability governs \emph{forensic reconstruction depth}. Because G2 replay depends on G3's sender-signed input (Property~\ref{prop:g2_input_integrity}), the effective verification depth of any execution path is $\min(\mathrm{CVD}, \mathrm{CAD})$---the minimum of the two chain depths.

\subsection{G3 as a Governance Contract}

Like G1 and G2, G3 defines \emph{what} properties interaction records must satisfy, not \emph{how} the ledger is implemented. The governance contract requires: (1)~every inter-agent and agent-human interaction produces a record with bilateral signatures and content commitments; (2)~records are linked in a tamper-evident chain; (3)~sufficient metadata exists for forensic reconstruction.

Our current instantiation uses a \textbf{centralized append-only ledger} operated by a trusted organizational entity (e.g., the deploying enterprise). This is the simplest design that satisfies G3: agents submit bilaterally-signed records to the operator, who maintains the hash chain and provides query access during audits. The operator is assumed honest for \emph{availability} (it will not selectively omit records) and \emph{ordering} (it will not reorder records), while \emph{integrity} is enforced cryptographically regardless of operator behavior (bilateral signatures + hash chain prevent undetectable modification).

Alternative instantiations satisfying the same G3 contract include: (a)~\emph{multi-operator consensus}, where $f{+}1$-of-$n$ operators must agree on each append, eliminating single-operator availability trust at the cost of consensus latency; (b)~\emph{append-only transparency logs} analogous to Certificate Transparency, where multiple independent monitors cross-check the ledger for consistency; (c)~\emph{blockchain-anchored commitments}, where periodic ledger root hashes are anchored to a public blockchain for non-repudiable timestamping. Each alternative trades latency or operational complexity for stronger trust assumptions. The G3 contract is invariant across all instantiations---what changes is the trust model of the storage layer, not the governance semantics of the records.

\section{Formal Security Properties}
\label{sec:proofs}

We state seven security properties and two theorems under the adversary model of \S\ref{sec:threat}. The basic instantiation uses standard cryptographic assumptions: EUF-CMA security of Ed25519, collision resistance of SHA-256, and the DiFR token divergence property~\cite{difr}. The enhanced instantiation additionally relies on $q$-SDH and DDH in Type-3 bilinear groups (BBS+ unforgeability and zero-knowledge, \S\ref{sec:bbs_plus}) and knowledge soundness of Groth16 (DV-SNARK, \S\ref{sec:dv_snark}). All proofs are in Appendix~\ref{app:proofs}.

\begin{property}[Trust Containment]
\label{thm:trust_containment}
For all edges $(u, v) \in E$: $\kappa_v \leq_\kappa \kappa_u$.
\end{property}

\begin{property}[Capability Binding]
\label{thm:capability_binding}
If $H_v(t') \neq H_v^{\mathit{cert}}$, then $\forall t \geq t': \mathrm{Access}(v, c, t) = \mathrm{DENY}$.
\end{property}

\begin{property}[Chain Integrity]
\label{thm:chain_integrity}
Only chains forming a valid path from a trusted root with valid signatures pass verification.
\end{property}

\begin{property}[Delegation Depth Enforcement]
\label{thm:depth}
Agent delegation depth from any non-agent ancestor never exceeds $\mathit{max\_depth}$.
\end{property}

\begin{property}[Reproducibility Soundness]
\label{thm:reproducibility}
Under software G2 (replay verification): if $v$ executes $M' \neq M$ or $S' \neq S$, replay verification detects with probability $\geq 1 - \epsilon(\lambda)$. Under hardware G2 (TEE): detection is deterministic.
\end{property}

\begin{theorem}[Bounded Divergence Under Verification]
\label{thm:bounded_divergence_proof}
If $M'$ passes replay verification (CharMatch $\geq \theta$) on $n$ independent prompts from $\mathcal{D}$, then $\Pr_{x \sim \mathcal{D}}[\mathrm{CharMatch}(M(x), M'(x)) < \theta] \leq \epsilon$ where $\epsilon \leq 1 - \alpha^{1/n}$, with $\alpha$ the significance level.
\end{theorem}

\begin{property}[Credential Isolation]
\label{thm:isolation}
$\mathrm{Accessible}(C') = \bigcup_{v \in C'} \mathrm{Accessible}(\{v\})$.
\end{property}

\begin{property}[Ledger Integrity]
\label{thm:ledger}
No party can undetectably modify, delete, or reorder records.
\end{property}

\begin{theorem}[Chain Verifiability]
\label{thm:chain_verifiability_proof}
Let $v_0 \to v_1 \to \cdots \to v_n$ be an execution chain where $v_0$ is a non-agent trust anchor, and let $\mathrm{CVD}(v_n) = k < n$ (Definition~\ref{def:chain_depth}). Then for all $i > k$, the output of $v_i$ cannot be verified as originating from a legitimate computation chain rooted at $v_0$, even if $\rho_{v_i}.\mathit{level} = \mathit{full}$.
\end{theorem}

\begin{property}[G2 Input Integrity via Bilateral Signing]
\label{prop:g2_input_integrity}
Replay verification for a compromised agent $v_j$ uses \emph{sender-authenticated} input, not self-reported input. Consequently, G2 soundness does not circularly depend on $v_j$'s own honesty.
\end{property}

\textbf{Rationale.}
A potential concern is circular dependency: G2 replay verification requires the recorded input to be authentic, but if the compromised agent $v_j$ controls its own input recording, it could fabricate inputs that make its unauthorized output appear legitimate. We show this is precluded by the bilateral signing structure of G3.

In any execution chain, agent $v_j$'s input is the output of its predecessor $v_{j-1}$. The interaction record $R_i$ in the ledger contains the input commitment $\eta_{\mathit{in}} = \mathrm{SHA\text{-}256}(\mathit{input})$ and is \emph{bilaterally signed} by both $v_{j-1}$ (sender) and $v_j$ (receiver). We consider three cases for $v_{j-1}$:

\emph{Case 1: $v_{j-1}$ is a non-agent node (human or organization).} By the adversary model (\S\ref{sec:threat}), non-agent nodes are uncompromised. $v_{j-1}$ signs the interaction record containing $\eta_{\mathit{in}}$. Agent $v_j$ cannot modify $\eta_{\mathit{in}}$ without invalidating $v_{j-1}$'s signature (EUF-CMA). Therefore the input committed to the ledger is the true input $v_j$ received.

\emph{Case 2: $v_{j-1}$ is an honest agent.} The same argument applies: $v_{j-1}$ signs its output and the corresponding ledger record. $v_j$ cannot forge $v_{j-1}$'s signature on a different $\eta_{\mathit{in}}$.

\emph{Case 3: $v_{j-1}$ is also compromised.} Then $v_{j-1}$'s output is itself unverifiable. However, the Chain Verifiability Theorem (Theorem~\ref{thm:chain_verifiability_proof}) already establishes that all nodes downstream of the \emph{first} unverifiable agent lose end-to-end verification guarantees. In this case, $v_j$ is downstream of the already-compromised $v_{j-1}$, so the chain verifiability depth $\mathrm{CVD} \leq j-1 < j$, and no claim of end-to-end verification extends to $v_j$ regardless.

In Cases~1 and~2, the verifier reconstructs $v_j$'s input via the forensic reconstruction protocol (\S\ref{sec:ledger}): request $v_{j-1}$ to disclose the original content, verify it against the sender-signed $\eta_{\mathit{in}}$, then replay $v_j$'s declared model on this \emph{authenticated} input. The dependency is therefore \textbf{unidirectional}---G2 relies on G3's sender-side input commitment, which is established independently of $v_j$---not circular.

\textbf{Chain Auditability.} This unidirectional dependency implies that G3 also exhibits a chain property: if an interaction record along an execution path is absent from the ledger, G2 replay verification at that point becomes impossible (the authenticated input is unavailable), and by induction all downstream provenance tracing fails. The effective verification depth of any execution path is therefore $\min(\mathrm{CVD}, \mathrm{CAD})$---the minimum of Chain Verifiability Depth (Definition~\ref{def:chain_depth}) and Chain Auditability Depth (Definition~\ref{def:cad}). This symmetric pair of chain properties---G2 governing behavioral verification depth, G3 governing forensic reconstruction depth---underscores the mutual reinforcement of the governance requirements.

\section{Implementation and Evaluation}
\label{sec:eval}

\subsection{Prototype Architecture}

We implement the framework as two Rust crates: \texttt{agent-trust-kit} (open-source SDK) and \texttt{a2auth} (credential governance application). Table~\ref{tab:implementation} summarizes the implementation.
\begin{table}[t]
\centering
\caption{Implementation summary.}
\label{tab:implementation}
\small
\resizebox{\columnwidth}{!}{%
\begin{tabular}{@{}llr@{}}
\toprule
\textbf{Component} & \textbf{Dependencies} & \textbf{LoC} \\
\midrule
Core SDK (cert, skills, trust tree, & \texttt{ring}, \texttt{sha2}, \texttt{ed25519-dalek}, & ${\sim}$5,600 \\
\quad ledger, vault, registry, MCP) & \texttt{sled}, \texttt{aes-gcm}, \texttt{clap} & \\
Enhanced G1: BBS+ disclosure & \texttt{zkryptium} & ${\sim}$270 \\
Enhanced G2: DV-SNARK circuit & \texttt{ark-groth16}, \texttt{ark-r1cs-std} & ${\sim}$350 \\
Real LLM repro. validation & \texttt{anthropic}, \texttt{sklearn} (Python) & ${\sim}$870 \\
Adversarial G2 validation & \texttt{transformers}, \texttt{peft} (Python) & ${\sim}$2,680 \\
MCP governance proxy demo & \texttt{cryptography}, \texttt{cbor2} (Python) & ${\sim}$520 \\
\midrule
\textbf{Total (Rust + Python)} & & ${\sim}$\textbf{11,090} \\
\bottomrule
\end{tabular}}
\end{table}

\subsection{Performance Evaluation}

We evaluate on Apple M2, 16\,GB RAM, macOS (Criterion.rs, 10,000 samples, 3s warmup). Table~\ref{tab:micro} summarizes microbenchmarks for the four core governance primitives.
\begin{table}[H]
\centering
\caption{Core governance microbenchmarks (Experiments~1--4).}
\label{tab:micro}
\small
\begin{tabular}{@{}llr@{}}
\toprule
\textbf{Operation} & \textbf{Key parameter} & \textbf{Latency} \\
\midrule
\multicolumn{3}{@{}l}{\emph{Exp 1: Certificate chain verification}} \\
\quad Depth-2 chain & 2 Ed25519 sigs & 66.1\,\textmu s \\
\quad Depth-5 chain (max practical) & 5 Ed25519 sigs & 163.6\,\textmu s \\
\quad Capability binding & 32-byte compare & $<$1\,ns \\
\midrule
\multicolumn{3}{@{}l}{\emph{Exp 2: Skills manifest hash}} \\
\quad 10 tools & Sort+CBOR+SHA-256 & 5.10\,\textmu s \\
\quad 100 tools & Sort+CBOR+SHA-256 & 51.6\,\textmu s \\
\midrule
\multicolumn{3}{@{}l}{\emph{Exp 3: Ledger operations}} \\
\quad Append record & Ed25519 sign+sled & 448\,\textmu s \\
\quad Verify single record & Ed25519 verify & 32.4\,\textmu s \\
\quad Full audit (100K records) & Chain walk & 3.46\,s \\
\midrule
\multicolumn{3}{@{}l}{\emph{Exp 4: Reproducibility verification}} \\
\quad Token compare (512 tokens) & Byte compare & 416\,ns \\
\quad Commit (512 tokens) & SHA-256 & 6.55\,\textmu s \\
\quad Batch verify (100$\times$512 tok) & & 41.8\,\textmu s \\
\bottomrule
\end{tabular}
\end{table}

Verification scales at ${\sim}$33\,\textmu s/sig; token comparison at ${\sim}$0.73\,ns/token (100K interactions audited in ${\sim}$42\,ms).

\textbf{Experiment 5: Real LLM Reproducibility Validation.} We validate the DiFR finding across \textbf{9 models from 7 providers} (Claude Sonnet~4/Haiku~4.5, GPT-4.1/4.1~mini, Gemini~2.5~Pro, Grok~4, DeepSeek~V3, Llama~4~Maverick, Kimi~K2.5) via OpenRouter with temperature$=$0, 100 prompts across six task types, 10 repetitions per model-prompt pair (5,400 API calls; 15,120 pairwise comparisons). We compute four similarity metrics---CharMatch, Jaccard, TF-IDF cosine, char n-gram cosine---and use Cohen's $d$ and Youden's $J$ for separation analysis and threshold selection.

\emph{Result 5b.} Table~\ref{tab:exp5b} shows multi-provider separation ($p \approx 0$ for all metrics).
\begin{table}[H]
\centering
\caption{Multi-provider separation: 9 models, 7 providers (15,120 comparisons).}
\label{tab:exp5b}
\small
\begin{tabular}{@{}lrrrr@{}}
\toprule
\textbf{Metric} & \textbf{Same} & \textbf{Cross} & \textbf{Sep.} & \textbf{$d$} \\
\midrule
CharMatch        & 0.538 & 0.114 & $4.7\times$ & 1.36 \\
Jaccard          & 0.710 & 0.321 & $2.2\times$ & 1.51 \\
TF-IDF cosine    & 0.839 & 0.586 & $1.4\times$ & 1.08 \\
Char n-gram cos  & 0.871 & 0.640 & $1.4\times$ & 1.20 \\
\bottomrule
\end{tabular}
\end{table}

All metrics achieve Cohen's $d > 1.0$; Jaccard ($d = 1.51$) and CharMatch ($d = 1.36$) provide the strongest separation.

\emph{Result 5c.} Table~\ref{tab:exp5c} ranks per-model determinism.
\begin{table}[H]
\centering
\caption{Per-model same-model reproducibility (CharMatch, temperature$=$0).}
\label{tab:exp5c}
\small
\begin{tabular}{@{}llrr@{}}
\toprule
\textbf{Model} & \textbf{Provider} & \textbf{CharMatch} & \textbf{Std} \\
\midrule
Claude Haiku 4.5   & Anthropic & 0.987 & 0.088 \\
Claude Sonnet 4    & Anthropic & 0.823 & 0.273 \\
GPT-4.1 mini       & OpenAI    & 0.701 & 0.335 \\
GPT-4.1            & OpenAI    & 0.617 & 0.363 \\
Llama 4 Maverick   & Meta      & 0.528 & 0.360 \\
Grok 4$^\dagger$   & xAI       & 0.402 & 0.343 \\
Gemini 2.5 Pro$^\dagger$ & Google & 0.392 & 0.397 \\
DeepSeek V3        & DeepSeek  & 0.225 & 0.271 \\
Kimi K2.5$^\dagger$ & Moonshot & 0.169 & 0.233 \\
\bottomrule
\multicolumn{4}{@{}l}{\footnotesize $^\dagger$Reasoning/thinking model (internal chain-of-thought).}
\end{tabular}
\end{table}

Determinism spans $5.8\times$ (Haiku~4.5: $0.987$ to Kimi~K2.5: $0.169$); \emph{even the least deterministic model} exceeds mean cross-provider similarity ($0.114$). Per-task separation ranges from $12.2\times$ (math) to $2.6\times$ (JSON).

\emph{Result 5e.} Via Youden's $J$ (Table~\ref{tab:exp5e}), Jaccard and CharMatch attain $F_1 > 0.87$ across all 9 models, confirming lightweight textual similarity provides practical model identity verification without model internals.
\begin{table}[H]
\centering
\caption{Optimal classification thresholds (multi-provider, 15,120 comparisons).}
\label{tab:exp5e}
\small
\begin{tabular}{@{}lrrr@{}}
\toprule
\textbf{Metric} & $\theta^*$ & \textbf{Youden's $J$} & $F_1$ \\
\midrule
Jaccard          & 0.408 & 0.592 & 0.876 \\
CharMatch        & 0.146 & 0.572 & 0.873 \\
Char n-gram cos  & 0.809 & 0.538 & 0.855 \\
TF-IDF cosine    & 0.837 & 0.466 & 0.844 \\
\bottomrule
\end{tabular}
\end{table}

\textbf{Experiment 6: Adversarial Model Substitution Detection.} We train three QLoRA backdoor variants (keyword B1, semantic B2, compositional B3) on three architectures (Qwen~2.5~7B, Llama~3.1~8B, Gemma~2~9B) and test G2 replay detection across four prompt strategies at budgets $n \in \{10, 25, 50, 100, 200\}$ (357 configs).
\begin{table}[H]
\centering
\caption{G2 adversarial detection rate at temperature${}=0$ (3 models $\times$ 4 strategies).}
\label{tab:adversarial}
\small
\begin{tabular}{@{}lrrrr@{}}
\toprule
\textbf{Budget ($n$)} & \textbf{B1} & \textbf{B2} & \textbf{B3} & \textbf{Avg} \\
\midrule
10  & 99.0\% & 99.2\% & 99.2\% & 99.1\% \\
25  & 100\%  & 100\%  & 100\%  & 100\%  \\
200 & 99.1\% & 98.9\% & 99.1\% & 99.1\% \\
\bottomrule
\end{tabular}
\end{table}

Detection is near-universal (minimum 90\% across all 357 configs). Mean CharMatch between base and backdoored outputs is 0.057, far below $\theta^* = 0.146$: modifying weights to implant an effective backdoor unavoidably shifts the clean-input distribution, making substitution detectable (\S\ref{sec:limitations}).

\subsection{Overhead and Comparison}
\textbf{Experiment 7: End-to-End MCP Overhead.}
\begin{table}[H]
\centering
\small
\begin{tabular}{@{}lr@{}}
\toprule
\textbf{Phase} & \textbf{Latency} \\
\midrule
MCP request parse & 0.02\,ms \\
Certificate chain verification & 0.097\,ms \\
Skills hash comparison & $<$0.001\,ms \\
Proxy token validation & 0.05\,ms \\
Ledger append & 0.448\,ms \\
\midrule
\textbf{Total governance overhead} & \textbf{${\sim}$0.62\,ms} \\
Typical MCP call (no governance) & 50--500\,ms \\
\textbf{Overhead as \% of total} & \textbf{0.1--1.2\%} \\
\bottomrule
\end{tabular}
\end{table}

\textbf{Experiment 8: Enhanced Instantiation Benchmarks.}
\begin{table}[H]
\centering
\caption{Enhanced cryptographic backends: BBS+ (G1) and Groth16 DV-SNARK (G2).}
\label{tab:bbs_snark_bench}
\small
\begin{tabular}{@{}lrr@{}}
\toprule
\textbf{Operation} & \textbf{Latency} & \textbf{vs.\ Basic} \\
\midrule
\multicolumn{3}{@{}l}{\emph{BBS+ selective disclosure (16-message cert)}} \\
\quad Full disclosure verification & 12.27\,ms & $127\times$ \\
\quad Selective proof gen (6/16) & 14.66\,ms & $151\times$ \\
\quad Selective proof verification & 13.80\,ms & $142\times$ \\
\midrule
\multicolumn{3}{@{}l}{\emph{Groth16 DV-SNARK CharMatch circuit}} \\
\quad Prove (1,024 chars, 19K constraints) & 173.86\,ms & --- \\
\quad Prove (4,096 chars, ${\sim}$73K constraints) & 750.33\,ms & --- \\
\quad Verify (constant) & 1.28\,ms & --- \\
\bottomrule
\end{tabular}
\end{table}

\textbf{Experiment 9: Comparison with BAID.}
\begin{table}[H]
\centering
\small
\resizebox{\columnwidth}{!}{%
\begin{tabular}{@{}lrrr@{}}
\toprule
\textbf{Metric} & \textbf{BAID} & \textbf{Ours (basic)} & \textbf{Ours (enhanced)} \\
\midrule
G1 verification & --- & 0.097\,ms & 13.80\,ms \\
G2 proof gen (1K) & ${\sim}$120\,s & ${\sim}$0.1\,ms$^\dagger$ & 174\,ms \\
G2 proof verify & ${\sim}$10\,ms & $<$0.01\,ms$^\dagger$ & 1.3\,ms \\
Circuit size & ${\sim}10^{10}$ & --- & 19,238 \\
Privacy & Full ZK & None & Selective / ZK \\
Infrastructure & Blockchain & None & None \\
\bottomrule
\multicolumn{4}{@{}l}{\footnotesize $^\dagger$Direct CharMatch comparison (no SNARK).}
\end{tabular}}
\end{table}

BBS+ adds ${\sim}142\times$ overhead but enables selective disclosure; DV-SNARK verification is constant at ${\sim}1.3$\,ms. Our DV-SNARK is ${\sim}690\times$ faster than BAID for proof generation (${\sim}520{,}000\times$ fewer constraints), by SNARKing only the verification metric rather than the full transformer forward pass. All enhanced operations remain within typical agent interaction latency (50--500\,ms).

\textbf{Experiment 10: MCP Governance Proxy Demo.} A Python proxy (520 lines) enforcing G1--G3 detects all 5 live attack scenarios (tool injection, trojanization, forged certificate, rate abuse, ledger tampering) in $\leq$165\,\textmu s each, achieving 8,195 governed calls/sec with mean overhead of 120\,\textmu s/call.

\subsection{Attack Scenario Evaluation}
We evaluate 12 attack scenarios at three levels of validation (Table~\ref{tab:attacks}): (1)~\emph{Structural governance attacks} (S1--S9) are tested as automated Rust unit tests that verify the governance algorithms correctly reject each attack construction---these validate \emph{algorithm correctness}, not resistance to real-world attackers. (2)~\emph{Runtime behavioral attacks} (S10--S11), informed by the Nemotron-AIQ dataset~\cite{ghosh_framework}, are validated via \emph{post-hoc forensic tracing}: the ledger provides sufficient evidence to identify the injection point after the fact, but the framework does not prevent the attack in real time. (3)~S12 (output exfiltration) is detectable via replay divergence. The 6 E2E attacks in \S\ref{sec:e2e} provide the strongest validation: real-time detection in a live pipeline with real LLM calls.
\begin{table*}[t]
\centering
\caption{Attack scenario evaluation. S1--S9: algorithm correctness validated via automated unit tests (9/9 pass). S10--S11: post-hoc forensic tracing validated against Nemotron-AIQ traces. S12: detectable via replay divergence. The ``Framework'' column distinguishes \emph{Prevented} (blocked in real time), \emph{Detected} (identified but not blocked), and \emph{Traced} (forensic reconstruction after the fact).}
\label{tab:attacks}
\small
\begin{tabular}{@{}llccl@{}}
\toprule
\textbf{Scenario} & \textbf{Attack} & \textbf{Baseline} & \textbf{Framework} & \textbf{Detection Mechanism} \\
\midrule
S1: Silent escalation & Add tool at runtime & Undetected & Prevented & Skills manifest hash mismatch \\
S2: Tool trojanization & Replace implementation & Undetected & Prevented & Code hash change \\
S3: Model substitution & Forged cert, wrong model & Undetected & Detected & Signature verification failure \\
S4: Phantom delegation & T2 issues T1 cert & Undetected & Prevented & Constraint monotonicity \\
S5: Evidence tampering & Modify ledger record & N/A & Detected & Hash chain + sig invalidation \\
S6: Blame shifting & Dispute in chain & Unresolvable & Resolved & Ed25519 bilateral signing \\
S7: Credential collusion & T3+T3 $\to$ T1 access & Undetected & Prevented & Independent tier evaluation \\
S8: Chain forgery & Fabricated chain & --- & Prevented & Sig verify: fake key $\neq$ root \\
S9: Depth overflow & Delegate at depth 0 & --- & Prevented & Depth counter exhausted \\
\midrule
S10: Knowledge poisoning & Inject via tool output & Undetected$^\dagger$ & Traced & Ledger $\eta_{\mathit{in}}/\eta_{\mathit{out}}$ \\
S11: Cascading hijack & Multi-stage persistence & Undetected$^\dagger$ & Traced & Hash-chain provenance audit \\
S12: Output exfiltration & Embed phishing link & Undetected$^\dagger$ & Detected & Replay divergence \\
\bottomrule
\multicolumn{5}{@{}l}{\footnotesize $^\dagger$100\% attack success in Nemotron-AIQ undefended traces (2,596 security trials).}
\end{tabular}
\end{table*}

\subsection{End-to-End Multi-Agent Pipeline}
\label{sec:e2e}

We deploy the full governance stack in a multi-agent pipeline with real Claude API calls, measuring overhead, scaling, attack detection, and baseline comparison (Figure~\ref{fig:e2e_pipeline}).
\begin{figure}[t]
  \centering
  \includegraphics[width=\columnwidth]{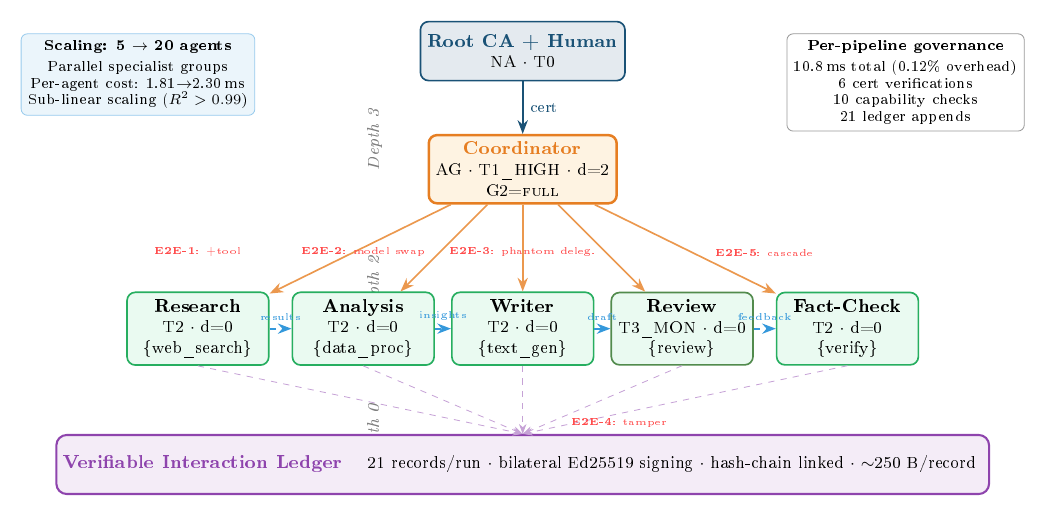}
  \caption{Multi-agent pipeline with governance overlay. Red labels indicate attack injection points (Table~\ref{tab:e2e_attacks}).}
  \label{fig:e2e_pipeline}
\end{figure}

\textbf{Pipeline architecture (Figure~\ref{fig:e2e_pipeline}).} We implement an AI Research Assistant using LangGraph with the Rust \texttt{agent-trust-kit} SDK exposed via PyO3 FFI. A Coordinator (T1\_HIGH) orchestrates four specialist agents---Research~(T2), Analysis~(T2), Writer~(T2), and Review~(T3\_MONITORED)---in a depth-3 trust tree. Each pipeline run produces 7 governed tool calls, each generating a bilaterally-signed ledger record with capability verification and reproducibility commitment.

\textbf{Experiment 11a: Governance overhead.}
\begin{table}[H]
\centering
\caption{Per-pipeline governance overhead (5-agent config, 7 ledger entries/run, mean of 90 runs with real Claude API calls).}
\label{tab:e2e_overhead}
\small
\begin{tabular}{@{}lrrr@{}}
\toprule
\textbf{Operation} & \textbf{Per-call} & \textbf{Calls} & \textbf{Total} \\
\midrule
G1: Capability check & 39\,\textmu s & 7 & 0.27\,ms \\
G2: Repro.\ commitment & 86\,\textmu s & 7 & 0.60\,ms \\
G3: Ledger append & 1,300\,\textmu s & 7 & 9.10\,ms \\
\midrule
\textbf{Total governance} & & \textbf{7} & \textbf{10.0\,ms} \\
\midrule
Pipeline latency (real Claude API) & & & ${\sim}$59.8\,s \\
\textbf{Overhead} & & & \textbf{$<$0.02\%} \\
\bottomrule
\end{tabular}
\end{table}

Governance adds ${\sim}$10\,ms to a ${\sim}$60\,s real-LLM pipeline ($<$0.02\%), dominated by G3 ledger appends (91\% of governance cost). Across 90 clean runs (630 governed tool calls), \textbf{zero false positives} and \textbf{100\% ledger integrity} are observed (95\% CI: FPR $<$ 3.3\%).

\textbf{Experiment 11b: Scaling behavior.}
\begin{table}[H]
\centering
\caption{Governance scaling: 5 $\to$ 20 agents (mock tool latency to isolate governance cost).}
\label{tab:e2e_scaling}
\small
\resizebox{\columnwidth}{!}{%
\begin{tabular}{@{}lrrrrr@{}}
\toprule
\textbf{Config} & \textbf{Agents} & \textbf{Ledger} & \textbf{Gov.} & \textbf{Per-agent} & \textbf{Storage} \\
 & & \textbf{entries} & \textbf{(ms)} & \textbf{(ms)} & \textbf{(KB)} \\
\midrule
Small & 5 & 7 & 6.9 & 1.37 & 1.7 \\
Medium & 10 & 16 & 12.9 & 1.29 & 3.9 \\
Large & 20 & 33 & 30.0 & 1.50 & 8.1 \\
\bottomrule
\end{tabular}}
\end{table}

Governance scales linearly with ledger entries; per-agent cost remains approximately constant (1.29--1.50\,ms across $4\times$ scaling). Storage is ${\sim}$250\,bytes/entry (8.1\,KB for 20 agents).

\textbf{Experiment 11c: End-to-end attack detection.}
\begin{table}[H]
\centering
\caption{End-to-end attack detection (real Claude API calls, 7 scenarios covering G1/G2/G3).}
\label{tab:e2e_attacks}
\small
\resizebox{\columnwidth}{!}{%
\begin{tabular}{@{}llll@{}}
\toprule
\textbf{Attack} & \textbf{Target} & \textbf{Layer} & \textbf{Detection} \\
\midrule
E2E-1: Cap.\ escalation & Research & G1 & Tool not in manifest \\
E2E-2: Model substitution & Analysis & G1 & Model binding mismatch \\
E2E-3: Phantom delegation & Writer & G1 & No valid cert chain \\
E2E-4: Ledger tampering & Ledger & G3 & Hash chain break \\
E2E-5: Tool trojanization & Research & G1 & Manifest hash mismatch \\
E2E-6: Depth overflow & Review & G1 & Trust decay violation \\
E2E-7: Replay divergence & Analysis & G2 & CharMatch 0.00 $<$ $\theta$ \\
\bottomrule
\end{tabular}}
\end{table}

All seven attacks are detected, covering all three governance layers. E2E-2 is caught by G1's model binding check: the certificate declares \texttt{claude-sonnet-4} but the agent uses \texttt{claude-haiku-4-5}. E2E-7 is the only attack requiring G2 (replay verification): the analysis agent's certificate and model binding are correct (G1 passes), but the agent secretly produces attacker-controlled output instead of executing the declared model. Post-pipeline replay re-executes the declared model on the stored input and detects divergence (CharMatch$=$0.00 vs.\ $\theta$$=$0.85)---demonstrating that G2 catches attacks invisible to G1 and G3.

\textbf{Experiment 11d: Baseline comparison.}
We evaluate the same seven attacks against three baselines.
\begin{table}[H]
\centering
\caption{E2E attack detection: A2Auth vs.\ baselines ($\checkmark$ = detected, \ding{55} = undetected).}
\label{tab:e2e_baseline}
\small
\resizebox{\columnwidth}{!}{%
\begin{tabular}{@{}lcccc@{}}
\toprule
\textbf{Attack} & \textbf{None} & \textbf{OAuth} & \textbf{OTEL} & \textbf{A2Auth} \\
\midrule
E2E-1: Cap.\ escalation & \ding{55} & \ding{55} & \ding{55} & $\checkmark$ (G1) \\
E2E-2: Model substitution & \ding{55} & \ding{55} & \ding{55} & $\checkmark$ (G1) \\
E2E-3: Phantom delegation & \ding{55} & \ding{55}$^a$ & \ding{55} & $\checkmark$ (G1) \\
E2E-4: Ledger tampering & N/A & N/A & \ding{55}$^b$ & $\checkmark$ (G3) \\
E2E-5: Tool trojanization & \ding{55} & \ding{55} & \ding{55} & $\checkmark$ (G1) \\
E2E-6: Depth overflow & \ding{55} & \ding{55} & \ding{55} & $\checkmark$ (G1) \\
E2E-7: Replay divergence & \ding{55} & \ding{55} & \ding{55} & $\checkmark$ (G2) \\
\midrule
\textbf{Total detected} & \textbf{0/7} & \textbf{0/7} & \textbf{0/7} & \textbf{7/7} \\
\bottomrule
\multicolumn{5}{@{}p{0.93\columnwidth}}{\footnotesize $^a$OAuth authenticates the agent but has no delegation depth constraint. $^b$OTEL traces can be silently modified (no hash chain or signatures).} \\
\end{tabular}}
\end{table}

OAuth~2.1 verifies \emph{who} the agent is but not \emph{what tools it currently has} or \emph{which model it runs}. Only A2Auth (G1+G2+G3) detects all seven E2E attacks, at $<$0.02\% overhead.

\textbf{Experiment 11e: Governance coverage comparison with BAID and Ghosh et al.}
Table~\ref{tab:e2e_baseline} compares against infrastructure baselines; Table~\ref{tab:governance_coverage} compares against the two closest \emph{research} systems---BAID~\cite{baid} (zkVM-based verifiable inference) and Ghosh et al.~\cite{ghosh_framework} (operational safety with Nemotron-AIQ)---at the governance requirement level across all 12 attack categories. This is a design-level analysis: we map each system's published mechanisms to the attack taxonomy of \S\ref{sec:threat}.

\begin{table*}[t]
\centering
\caption{Governance coverage: A2Auth vs.\ BAID and Ghosh et al.\ across 12 attack categories. $\CIRCLE$ = prevented/detected in real time; $\LEFTcircle$ = partially addressed or post-hoc only; \ding{55} = not addressed. Analysis is based on each system's published design; BAID and Ghosh et al.\ were not re-implemented.}
\label{tab:governance_coverage}
\small
\begin{tabular}{@{}llcccp{4.2cm}@{}}
\toprule
\textbf{Attack} & \textbf{Requires} & \textbf{BAID} & \textbf{Ghosh} & \textbf{A2Auth} & \textbf{Notes} \\
\midrule
ATK-1: Capability escalation & G1 & \ding{55} & $\LEFTcircle$ & $\CIRCLE$ & BAID verifies computation, not enrollment. Ghosh has policy scoping but no cryptographic binding. \\
ATK-2: Model substitution & G2 & $\CIRCLE$ & \ding{55} & $\CIRCLE$ & BAID provides deterministic proof via zkVM; A2Auth provides statistical detection via replay. \\
ATK-3: Tool trojanization & G1 & \ding{55} & \ding{55} & $\CIRCLE$ & Neither BAID nor Ghosh hashes tool implementations. \\
ATK-4--7: Delegation attacks & G1 & \ding{55} & \ding{55} & $\CIRCLE$ & No trust tree or delegation constraints in either system. \\
ATK-8: Evidence tampering & G3 & \ding{55} & $\LEFTcircle$ & $\CIRCLE$ & Ghosh uses OTEL (no hash chain); BAID has no ledger. \\
ATK-9: Blame shifting & G2+G3 & $\LEFTcircle$ & $\LEFTcircle$ & $\CIRCLE$ & BAID proves per-execution integrity but lacks bilateral signing for attribution. \\
ATK-10--11: Injection/hijack & G3 (trace) & \ding{55} & $\LEFTcircle$ & $\LEFTcircle$ & Ghosh traces via OTEL (no integrity); A2Auth traces via hash-chain ledger. Neither prevents in real time. \\
ATK-12: Output exfiltration & G2 & $\CIRCLE$ & \ding{55} & $\LEFTcircle$ & BAID proves correct execution deterministically; A2Auth detects via replay divergence. \\
\midrule
\multicolumn{2}{@{}l}{\textbf{Governance requirements}} & G2 only & Partial G3 & G1+G2+G3 & \\
\multicolumn{2}{@{}l}{\textbf{G2 proof generation}} & ${\sim}$120\,s & --- & ${\sim}$0.1\,ms$^*$ / 174\,ms$^\dagger$ & \\
\multicolumn{2}{@{}l}{\textbf{Infrastructure}} & Blockchain & None & None & \\
\bottomrule
\multicolumn{6}{@{}p{0.95\textwidth}}{\footnotesize $^*$Basic instantiation (direct CharMatch). $^\dagger$Enhanced instantiation (DV-SNARK). BAID provides \emph{computational integrity} (the declared model provably executed the computation)---a strictly stronger G2 guarantee than A2Auth's \emph{statistical behavioral verification}. However, BAID addresses only G2: it cannot detect capability escalation (G1) or audit interaction chains (G3). Ghosh et al.'s operational framework provides guardrails and red-teaming methodology but lacks cryptographic enforcement for any governance requirement. The three systems are complementary, not competing: BAID's zkVM could serve as an enhanced G2 backend within the A2Auth architecture; Ghosh's guardrails address content safety orthogonal to governance.} \\
\end{tabular}
\end{table*}

\section{Limitations, Governance Trilemma, and Future Work}
\label{sec:limitations}

\subsection{Scope of Contribution}

Our contribution is a governance \emph{architecture}, not new cryptographic primitives. The building blocks (digital signatures, hash chains, replay verification) are individually well-known; the contribution lies in identifying the capability-context separation as a principled governance boundary, proving structural constraints (Chain Verifiability, Bounded Divergence) that arise from this architecture, and defining abstract governance contracts (G1--G3) that hold across fundamentally different cryptographic backends.

\textbf{Why the composition is non-obvious.} A counterfactual test: given Ed25519, SHA-256, and hash chains, would a security engineer independently arrive at G1--G3? The evidence suggests not. The MCP ecosystem has grown to over 1,000 tool servers with OAuth~2.1 as its security mechanism, yet no deployment detects capability drift (G1 gap). OpenTelemetry is the industry standard for agent observability, yet the Nemotron-AIQ dataset~\cite{ghosh_framework} demonstrates that 10,796 OTEL traces provide no tamper evidence (G3 gap). BAID~\cite{baid} constructs a zkVM over the full transformer pass, yet cannot detect whether the agent's \emph{tool set} has changed (G1 gap). Each existing approach addresses a fragment; none recognizes the capability-context separation as the organizing principle, and none arrives at the three-requirement conjunction that Table~\ref{tab:governance_gap} shows is necessary and sufficient. The \emph{composition} and the \emph{boundary identification} are the contributions, validated by two instantiations (basic and enhanced) satisfying the same nine security properties under standard assumptions (CR, EUF-CMA; $q$-SDH, DDH, knowledge soundness).

\subsection{The Agent Governance Trilemma}
\label{sec:trilemma}

A fundamental question underlies the design of any agent governance framework: \emph{why can't we simply have agents that are maximally capable, maximally fast, and maximally secure?} We observe that this is not merely a current engineering limitation, but a structural constraint rooted in computability theory.

\begin{observation}[Agent Governance Trilemma, informal]
\label{obs:trilemma}
For agent systems operating under a fixed computational budget, three desiderata---\emph{capability}, \emph{performance}, and \emph{security}---cannot be simultaneously maximized. Formally:
\begin{enumerate}[nosep]
  \item \textbf{High Capability + High Security $\Rightarrow$ Performance $\leq 1/2$.} If an agent's tool set is Turing-complete (unbounded capability) and full behavioral verification is required (high security), then verification costs at least as much as the original execution, so at most half the computational budget produces useful work.
  \item \textbf{Bounded Capability + High Security $\Rightarrow$ Performance $\to 1$.} If the tool set is non-Turing-complete (bounded capability), verification reduces to lookup in a pre-computed table of valid input-output pairs, with $O(1)$ per-query overhead.
  \item \textbf{High Capability + High Performance $\Rightarrow$ Security degrades.} If the agent is Turing-complete and most of the budget goes to useful work, insufficient resources remain for verification, and the verifier cannot detect all semantic deviations.
\end{enumerate}
\end{observation}

\textbf{Proof sketch for (1).} Behavioral verification (G2) requires determining whether the agent's output $y$ matches what its declared model $M$ would produce on input $x$. For general computations, the only universal method is \emph{re-execution}: compute $M(x)$ and compare. No sub-linear shortcut exists for arbitrary $M$, because $M$'s output depends on the full complexity of its parameters and computation---the function $M$ may be incompressible (its Kolmogorov complexity approaches its execution cost). Therefore $W_{\mathrm{verify}} \geq W_{\mathrm{task}}$, giving performance ratio $P = W_{\mathrm{task}} / (W_{\mathrm{task}} + W_{\mathrm{verify}}) \leq 1/2$. For Byzantine-fault-tolerant verification with $k$ independent verifiers, $P \leq 1/(k+1)$.

\textbf{Proof sketch for (2).} When each tool $t_i$ has a finite input space $\mathcal{I}_i$ and output space $\mathcal{O}_i$, the set of valid $(x, t_i(x))$ pairs is bounded by $\sum_i |\mathcal{I}_i|$. A hash table of valid pairs can be pre-computed offline and queried in $O(1)$ per verification. Since the pre-computation is amortized over all future queries, the marginal verification cost is constant, and $P \to 1$ as the number of task executions grows.

\textbf{Proof sketch for (3).} This follows from Rice's theorem: ``the output of agent $A$ is safe'' is a non-trivial semantic property of $A$'s program, and no algorithm can decide all such properties for Turing-complete programs. When the verification budget $W_{\mathrm{verify}} < W_{\mathrm{task}}$, the verifier cannot re-execute all computations. Sampling-based verification (checking $n$ out of $N$ outputs) provides statistical guarantees bounded by our Bounded Divergence Theorem ($\epsilon \leq \ln(1/\alpha)/n$), but the adversary can always deviate on the $N - n$ unchecked outputs. Static analysis of agent programs is undecidable in the general case.

\textbf{How TEE sidesteps the trilemma.} Hardware TEE replaces computational verification with hardware attestation, shifting the tradeoff from the compute domain to the trust assumption domain (reliance on hardware vendor integrity and supply chain security). TEE does not break the trilemma but trades compute cost for trust assumptions, and satisfies only G2---a TEE-attested agent still requires G1 and G3.

\textbf{Analogy to established impossibility results.} The trilemma is structurally analogous to the CAP theorem~\cite{cap_theorem} and the blockchain trilemma; we conjecture that a formal proof reducing to standard computability results is possible (\S\ref{sec:future}).

\subsection{Governance Levels: Operationalizing the Trilemma}
\label{sec:governance_levels}

The trilemma motivates three governance levels (Level~3: compile-time safety with bounded tools; Level~2: sampled verification for enterprise deployments; Level~1: post-hoc traceability for low-risk use cases), each corresponding to a different position in the capability-performance-security space. Full specifications including recommended crypto backends are given in Appendix~\ref{app:governance_levels}. The key insight is that every deployment's \emph{security posture is legible}: the governance level is declared in the agent's certificate and recorded in the ledger, making the tradeoff explicit to auditors and downstream agents.

\subsection{Capability Boundedness and Meta-Tools}
\label{sec:capability_boundedness}

The trilemma's capability dimension distinguishes two tool classes. \textbf{Bounded capabilities} (\texttt{web\_search}, \texttt{read\_file}) have finite or structurally constrained input-output spaces; G1 fully characterizes the envelope and verification is efficient. \textbf{Unbounded capabilities} (\texttt{code\_execution}, \texttt{shell\_access}) are Turing-complete: G1 correctly records that the agent \emph{has} the capability, but the envelope is \emph{formally declared yet semantically unbounded}---knowing an agent can execute code does not predict what code it will execute.

For unbounded tools, governance either restricts to bounded mode (sandboxed execution, Level~3) or accepts unbounded mode with runtime monitoring ($P \leq 1/2$ per the trilemma, as the monitor consumes at least as much compute as the governed agent). G1--G3 provides the infrastructure for both: G1 declares boundedness, G2 verifies behavioral fidelity within the envelope, and G3 records all interactions for post-hoc audit.

\subsection{Capability Binding Limitations}

\textbf{Closed-source binding.} For closed-source models and tools, the skills hash captures configuration metadata (name, version, API schema), not implementation code. A provider could change implementation without changing the public descriptor. Reproducibility verification provides a complementary detection mechanism: implementation changes that alter behavior will be detected via replay divergence.

\textbf{Re-signing frequency.} Mitigated through wildcard tool categories (\texttt{web\_search@\^{}2.0}) and short-lived auto-renewable certificates. A typical deployment (10 tools, weekly updates) incurs ${\sim}520$ certificate operations per year at $<$100\,\textmu s each---operationally negligible.

\subsection{Graceful Degradation under Partial G2}

When G2 is unavailable (reproducibility commitment $= \mathit{none}$), the framework degrades to G1+G3 only. Table~\ref{tab:degradation} summarizes the security impact.

\begin{table}[H]
\centering
\caption{Security guarantees under G1+G2+G3 vs.\ G1+G3 degradation.}
\label{tab:degradation}
\small
\begin{tabular}{@{}lcc@{}}
\toprule
\textbf{Attack class} & \textbf{G1+G2+G3} & \textbf{G1+G3} \\
\midrule
Silent capability escalation (ATK-1) & Prevented & Prevented \\
Tool trojanization (ATK-3) & Prevented & Prevented \\
Phantom delegation (ATK-4) & Prevented & Prevented \\
Chain/depth violation (ATK-5,6) & Prevented & Prevented \\
Evidence tampering (ATK-8) & Detected & Detected \\
Blame shifting (ATK-9) & Resolved & Resolved \\
Credential collusion (ATK-7) & Prevented & Prevented \\
\midrule
Model substitution (ATK-2) & \textbf{Detected} & \ding{55} \\
Output exfiltration (ATK-12) & \textbf{Detected} & \ding{55} \\
\bottomrule
\end{tabular}
\end{table}

G1+G3 alone protect against 9 of 12 attack types. The attacks requiring G2---model substitution (ATK-2), cascading injection via model compromise (ATK-10/11), and output exfiltration (ATK-12)---share a common structure: the attacker operates \emph{within} a valid capability envelope by substituting the computational process itself.

\textbf{Chain-level degradation.} A G2$=$\textsc{none} agent at an interior position breaks behavioral verification for all downstream nodes (Chain Verifiability Theorem, \S\ref{sec:chain_verifiability}); the deployment principles therein provide prescriptive architectural guidance.

\textbf{Statistical vs.\ computational security.} Token-level reproducibility provides statistical rather than computational security, but the Bounded Divergence Theorem (\S\ref{sec:indistinguishability}) bounds this gap: an adversary passing $n$ replay checks has produced a model $(n, \epsilon)$-indistinguishable from the original, with $\epsilon \leq \ln(1/\alpha)/n$.

\textbf{Verification metric limitations.} Current textual similarity metrics (CharMatch, Jaccard) have lower separation power for reasoning models with high paraphrase variation. The architecture is metric-agnostic: any similarity function satisfying the threshold separation property can be substituted, and semantic similarity (embedding cosine distance) is a natural extension.

\textbf{Provider cooperation.} Basic G2 requires access to the same model version, mitigated by negotiated replay rights, deterministic inference providers~\cite{eigenai}, or the enhanced G2 instantiation (DV-SNARK, \S\ref{sec:dv_snark}) which reduces cooperation to one-time proof generation.

\subsection{Trust Infrastructure Limitations}

\textbf{Non-agent node compromise.} We assume non-agent nodes (CAs, humans, organizations) are uncompromised. Standard PKI mitigations apply (CT logs, TPM~2.0 trust anchors, multi-party issuance, short certificate lifetimes); our framework inherits these risks and mitigations from the PKI foundation.

\textbf{Ledger operator trust.} The current G3 instantiation assumes a centralized operator honest for availability and ordering; record \emph{integrity} is enforced cryptographically (bilateral signatures, hash chains) regardless. Alternative instantiations (multi-operator consensus, transparency logs, blockchain anchoring) can relax the availability assumption at the cost of latency, preserving G3 semantics (\S\ref{sec:ledger}).

\textbf{Application-layer data flow.} Our framework governs agent \emph{identity and capability}, not the \emph{content} of data flows. Complementary approaches (data-flow taint tracking, differential privacy) address this orthogonal concern.

\subsection{Future Work}

\label{sec:future}
\textbf{Formalizing the trilemma.} A rigorous formalization reducing to Rice's theorem and simulation bounds---analogous to the Gilbert--Lynch proof of CAP (2002)---would establish the trilemma as a foundational constraint on agent governance.

\textbf{Hardware-attested identity.} TPM~2.0 and TEE~\cite{sgx_explained,arm_trustzone} integration for stronger G2 guarantees without provider cooperation, addressing the trilemma by substituting compute-domain verification with trust-domain attestation.

\textbf{Stealthy model substitution.} Our adversarial experiment (Experiment~6) shows QLoRA-based backdoors are trivially detected. More sophisticated attacks (distillation-based backdoors, gradient-masked training) may produce stealthier substitutions; characterizing the fundamental limits of replay-based detection is an open question.

\textbf{Semantic verification metrics.} Embedding-based semantic similarity for reproducibility verification with formal threshold analysis, potentially extending G2's applicability to low-determinism reasoning models.

\textbf{IVC for G3.} An enhanced G3 using Incrementally Verifiable Computation (e.g., Nova folding) would enable $O(1)$ audit verification regardless of chain length, deferred as current scales do not justify the ${\sim}200\times$ append overhead.

\textbf{Cross-protocol governance.} Extending capability binding beyond MCP/A2A to emerging agent protocols.

\textbf{Dynamic trust calibration.} Trust scores evolving based on ledger history, enabling governance levels to adapt automatically based on observed behavior.

\section{Related Work}
\label{sec:related}

We organize related work along the four-layer landscape.

\textbf{Identity Platforms.} Keyfactor~\cite{keyfactor} and HID Global~\cite{hid} propose X.509 for agents but without tool binding. SPIFFE/SPIRE~\cite{spiffe_hashicorp,spiffe_solo} provides workload identity. CyberArk~\cite{cyberark} extends PAM to agents.

\textbf{Authorization Frameworks.} OIDC-A~\cite{oidc_a} extends OpenID Connect with agent claims. South et al.~\cite{south_delegation} propose Agent-ID Tokens (ICML 2025). Agentic JWT~\cite{agentic_jwt} defines delegation with temporal constraints.

\textbf{Runtime Gateways.} Gravitee~\cite{gravitee} and Lasso~\cite{lasso} provide MCP proxies with per-call enforcement---fundamentally stateless.

\textbf{Governance Platforms.} Credo AI~\cite{credo_ai} maintains agent registries at the documentation layer.

\textbf{Verifiable Execution.} BAID~\cite{baid} is a high-assurance instantiation of the G2 governance requirement (\S\ref{sec:reproducibility}) within our framework: it uses full-inference zkVM SNARKs to achieve computational integrity without TEE hardware, cryptographically binding agent execution to a code commitment $C_P$, thereby satisfying G2 and partially G1 via code-level identity. However, BAID's $C_P$ does not encompass dynamically acquired tools---an agent's program binary remains unchanged when new tools are discovered via MCP or plugin registries, leaving the capability envelope outside its security perimeter (\S\ref{sec:intro_landscape}). Our framework is complementary rather than competing: BAID provides no G3 guarantee (no tamper-evident cross-agent audit trail) and incomplete G1 coverage for dynamic capabilities, whereas our framework delivers the complete G1$\wedge$G2$\wedge$G3 governance architecture with BAID serving as one pluggable G2 backend for deployments requiring maximum cryptographic assurance. For deployments prioritizing practicality, our DV-SNARK path (\S\ref{sec:dv_snark}) achieves privacy-preserving behavioral verification at ${\sim}690\times$ lower cost (${\sim}174$\,ms vs.\ ${\sim}120$\,s) by SNARKing only the \emph{verification metric} (${\sim}10^4$ constraints) rather than the full model forward pass (${\sim}10^{10}$ constraints); both paths satisfy the same G2 governance contract with different cost-assurance tradeoffs. DiFR~\cite{difr} demonstrates inference near-determinism. EigenAI~\cite{eigenai} achieves 100\% deterministic inference. Attestable Audits~\cite{attestable} uses TEEs~\cite{sgx_explained} for verifiable AI benchmarks.

\textbf{Classical Trust Management.} PolicyMaker~\cite{policymaker}, KeyNote~\cite{keynote}, SPKI/SDSI~\cite{spki}, F-PKI~\cite{fpki}, ARPKI~\cite{arpki}.

\textbf{Agent Protocols.} A2A~\cite{a2a_spec} provides agent communication. MCP~\cite{mcp_spec} enables agent-tool interaction with OAuth~2.1. Two recent works address MCP-specific security gaps that validate our G1 requirement. ETDI~\cite{etdi} extends MCP with cryptographic signing of tool definitions and immutable versioning, directly mitigating tool squatting and rug-pull attacks---an MCP-specific instantiation of capability integrity. MCPSec~\cite{mcpsec} identifies the \emph{absence of capability attestation} as a critical protocol vulnerability, a finding that independently corroborates our capability-identity gap analysis (\S\ref{sec:intro_landscape}). Both works are protocol-specific engineering solutions: neither articulates the capability-context separation as a general principle, addresses behavioral verifiability (G2), nor provides interaction auditability (G3). Our framework subsumes their observations as special cases of G1 while providing the complete G1$\wedge$G2$\wedge$G3 governance architecture.

\textbf{Agent Governance Frameworks.} Cheng et al.~\cite{cheng_safe} propose a three-pillar model (transparency, accountability, trustworthiness) as a conceptual governance framework. Allegrini et al.~\cite{allegrini_formal} formalize 31 task-lifecycle properties in CTL/LTL but address operational correctness rather than cryptographic governance. Ghosh et al.~\cite{ghosh_framework} propose an operational safety framework with agent red-teaming probes (ARP), policy-based tool scoping, and OpenTelemetry-based audit trails, validated through a 10,796-trace dataset (Nemotron-AIQ) covering 22 security attack templates and 13 content-safety categories against NVIDIA's AI-Q research assistant. Their empirical results demonstrate that \emph{without} guardrails, all 2,596 security attacks succeed (100\%), and 20\% of 2,800 safety attacks propagate to the final output---providing compelling evidence for the inadequacy of relying solely on model alignment. However, their OTEL-based audit trails lack cryptographic integrity: span linkage is purely referential (no hash chaining), spans carry no agent identity binding (no certificates), and full plaintext I/O is stored (no privacy preservation). Our framework addresses these gaps with a formally verified, hash-chain-linked, bilaterally signed interaction ledger (Section~\ref{sec:ledger}). Furthermore, their ``signed schemas'' and ``per-call capability tokens'' remain design principles without cryptographic specification or formal proofs. Raza et al.~\cite{trism} adapt TRiSM for agentic multi-agent systems.

\textbf{Instruction-Data Separation.} A related line of work addresses the model's inability to distinguish instructions from data \emph{within} the transformer. The Instruction Hierarchy~\cite{instruction_hierarchy} trains models to prioritize system prompts over user inputs and tool outputs, establishing trust levels per input type. However, this operates at the \emph{model} level via fine-tuning---it does not address the \emph{orchestration-layer} problem that our capability-context separation targets: tool definitions and runtime context have different governance lifecycles (infrequent re-authorization vs.\ per-interaction auditing) that no amount of model training can enforce. Our separation is complementary: even a perfectly instruction-following model requires external cryptographic infrastructure to detect when its capability envelope has changed (G1), verify its computational process (G2), and produce tamper-evident interaction records (G3).

\textbf{Regulatory Landscape.} The EU AI Act (Regulation 2024/1689) requires traceability but lacks agent-specific implementing acts~\cite{hannecke}. NIST CAISI~\cite{nist_caisi} is developing agent standards. CSA~\cite{csa_vouches,csa_trust} proposes agentic trust frameworks.

\section{Conclusion}
\label{sec:conclusion}

We identified \emph{silent capability escalation} as a new vulnerability class in AI agent ecosystems and traced its root cause to a missing architectural distinction: the \textbf{capability-context separation}, which recognizes that tool definitions and runtime context have fundamentally different security semantics at the orchestration layer. From this principle, we derived three Agent Governance Requirements (G1--G3) that define \emph{what} an agent ecosystem must enforce, independent of \emph{how}---a governance architecture, not new cryptographic primitives.

Two structural results have prescriptive architectural consequences. The \textbf{Chain Verifiability Theorem} establishes that behavioral verification is a chain property: one unverifiable interior agent breaks end-to-end verification for all downstream nodes, constraining how multi-agent systems must be architected. The \textbf{Bounded Divergence Theorem} transforms replay-based verification into a probabilistic safety certificate, establishing software G2 as a deployable approximation of trusted execution. A multi-provider reproducibility study (9 models, 7 providers) reveals $5.8\times$ variance in inference determinism, directly connecting model characteristics to governance architecture and elevating reproducibility to a first-class model performance metric.

We validated the framework with two crypto-agnostic instantiations---basic (Ed25519, SHA-256, hash chains; 97\,\textmu s verify) and enhanced (BBS+ selective disclosure, DV-SNARK; 13.8\,ms verify)---both satisfying the same nine security properties. End-to-end evaluation over 5--20 agent pipelines with real LLM calls confirms $<$0.02\% governance overhead and real-time detection of 7 end-to-end attack scenarios (covering G1, G2, and G3) with zero false positives, while three baselines (no governance, OAuth~2.1, OpenTelemetry) detect none. An additional 9 structural attacks are validated via automated unit tests, and 3 runtime behavioral attacks are validated via post-hoc forensic tracing.

The Agent Governance Trilemma (\S\ref{sec:trilemma}) reveals why no single configuration suffices: agents with unbounded capabilities (code execution, tool composition) require verification at least as expensive as the original execution, forcing an explicit tradeoff between capability, performance, and security. G1--G3 provides the infrastructure for the full governance spectrum---from compile-time safety (bounded tools, full verification) to post-hoc traceability (unbounded tools, audit-only)---making each deployment's security posture explicit and auditable.

As the EU AI Act's high-risk provisions take effect and NIST develops agent-specific standards, the need for cryptographic infrastructure enforcing G1--G3 will only grow. By making reproducibility a prerequisite for behavioral verifiability, our framework creates a concrete incentive for model providers to self-declare inference determinism characteristics---elevating reproducibility from an implementation detail to a first-class, security-critical model specification. We release the core mechanisms as an open-source SDK (\texttt{agent-trust-kit}) to enable ecosystem adoption.

\bibliographystyle{plain}

\appendix
\section{Proofs of Security Properties}
\label{app:proofs}

This appendix contains the full proofs for all theorems stated in \S\ref{sec:proofs}.

\begin{proof}[Proof of Property~\ref{thm:trust_containment} (Trust Containment)]
Suppose $\mathcal{A}$ issues $\mathit{cert}_w$ with $\kappa_w \not\leq_\kappa \kappa_v$. The verifier checks $\kappa_w \leq_\kappa \kappa_v$ using constraints in $\mathit{cert}_v$, which is signed by $v$'s parent. $\mathcal{A}$ cannot modify $\mathit{cert}_v$ without the parent's key. Recursing upward, the chain terminates at an uncompromised non-agent ancestor (EUF-CMA). \qed
\end{proof}

\begin{proof}[Proof of Property~\ref{thm:capability_binding} (Capability Binding)]
By collision resistance: $\Pr[H_v(t') = H_v^{\mathit{cert}}] \leq \mathrm{negl}(\lambda)$ when $S_v(t') \neq \sigma_v$. Forging a new certificate requires the parent's key (same argument as Property~\ref{thm:trust_containment}). \qed
\end{proof}

\begin{proof}[Proof of Property~\ref{thm:chain_integrity} (Chain Integrity)]
Fabricating any link requires the corresponding private key. By EUF-CMA, $\mathcal{A}$ cannot produce valid signatures without the key. \qed
\end{proof}

\begin{proof}[Proof of Property~\ref{thm:depth} (Delegation Depth Enforcement)]
Trust Containment requires $\kappa_v.\mathit{max\_depth} < \kappa_{\mathit{par}(v)}.\mathit{max\_depth}$ (strict). Starting from depth $d$, the chain has at most $d$ agent levels before reaching 0. \qed
\end{proof}

\begin{proof}[Proof of Property~\ref{thm:reproducibility} (Reproducibility Soundness)]
\emph{Hardware G2}: TEE attestation directly proves that the declared model was executed; substitution is detected with certainty. \emph{Software G2}: By DiFR~\cite{difr}, for different models the empirical match rate is $\delta \approx 0.06$ (character-level), validated in our Experiment~5 across 9 models and 7 providers. Our empirical data shows that the fraction of cross-model output pairs exceeding any reasonable threshold $\theta$ (e.g., $\theta = 0.15$) is at most $q \approx 0.02$ (i.e., $\Pr[\mathrm{CharMatch}(M(x), M'(x)) \geq \theta \mid M' \neq M] \leq q$). Each verification prompt is therefore an independent Bernoulli trial with pass probability at most $q$. The probability that a substituted model passes all $n$ verification prompts is at most $q^n$, which is negligible for realistic $n$ (e.g., $q^{10} \leq 10^{-17}$). \qed
\end{proof}

\begin{proof}[Proof of Theorem~\ref{thm:bounded_divergence_proof} (Bounded Divergence Under Verification)]
Let $p = \Pr_{x \sim \mathcal{D}}[\mathrm{CharMatch}(M(x), M'(x)) < \theta]$ be the true failure rate. Each verification trial $Z_i = \mathbf{1}[\mathrm{CharMatch}(M(x_i), M'(x_i)) \geq \theta]$ is an independent Bernoulli$(1-p)$ random variable. Observing all $n$ trials pass ($\sum Z_i = n$), we seek an upper confidence bound on $p$.

The probability of observing $n$ successes when the true failure rate is $p$ is $(1-p)^n$. The exact one-sided Clopper-Pearson upper bound at confidence level $1-\alpha$ is the value $p^*$ satisfying $(1-p^*)^n = \alpha$, giving $p^* = 1 - \alpha^{1/n}$.

For small $p^*$, using the approximation $\ln(1-p) \approx -p$: from $(1-p)^n = \alpha$ we get $n \ln(1-p) = \ln \alpha$, hence $p \approx \ln(1/\alpha)/n$.

This is a distribution-free result: the bound depends only on the number of verification trials $n$ and the desired confidence level $\alpha$, not on the threshold $\theta$ or any distributional assumptions on the CharMatch statistic. The role of $\theta$ is solely to define what constitutes a ``pass''---a higher $\theta$ increases sensitivity but may also increase false positives for the legitimate model. \qed
\end{proof}

\begin{proof}[Proof of Property~\ref{thm:isolation} (Credential Isolation)]
Access depends solely on the requesting agent's certificate. No protocol mechanism combines permissions. Proxy tokens are bound to specific agent IDs. \qed
\end{proof}

\begin{proof}[Proof of Property~\ref{thm:ledger} (Ledger Integrity)]
(a) Modification invalidates $h_{\mathit{prev}}$ of successor, cascading through chain; re-signing requires all subsequent senders' keys. (b) Deletion breaks hash chain. (c) Reordering breaks both sequence numbers and hash links. \qed
\end{proof}

\begin{proof}[Proof of Theorem~\ref{thm:chain_verifiability_proof} (Chain Verifiability)]
We proceed by induction on chain position.

\emph{Base case} ($i = k+1$): Agent $v_k$ has $\rho.\mathit{level} = \mathit{none}$. Under the adversary model (\S\ref{sec:threat}), $\mathcal{A}$ may have replaced $v_k$'s model $M_k$ with $M_k'$. Since no replay verification applies to $v_k$, the verifier cannot distinguish $v_k$'s actual output $o_k' = M_k'(\mathit{input}_k)$ from the legitimate output $o_k^* = M_k(\mathit{input}_k)$. Both are validly signed by $v_k$ (G3$\checkmark$) and $v_k$'s certificate is valid (G1$\checkmark$). Agent $v_{k+1}$ receives $o_k'$ as input. Even with $\rho_{v_{k+1}}.\mathit{level} = \mathit{full}$, replay verification confirms $v_{k+1}(o_k') = o_{k+1}$---i.e., $v_{k+1}$ faithfully processed $o_k'$. But since $o_k'$ may differ from $o_k^*$, the computation is correct on potentially corrupted input. The verifier cannot detect this because verifying $v_{k+1}$'s replay requires $v_k$'s output as the reference input, which is itself unverifiable.

\emph{Inductive step}: Assume the output of $v_{i-1}$ ($i-1 > k$) cannot be verified as legitimate. Agent $v_i$ receives this unverifiable output as input. By the same argument as the base case, $v_i$'s replay verification confirms computational fidelity on unverifiable input. \qed
\end{proof}

\begin{proof}[Proof of Property~\ref{prop:g2_input_integrity} (G2 Input Integrity via Bilateral Signing)]
Let $v_j$ be a compromised agent and let $R_i = (\ldots, a_{\mathit{send}}, a_{\mathit{recv}}, \ldots, \eta_{\mathit{in}}, \ldots, \mathit{sig}_i)$ be the ledger record for the interaction delivering input to $v_j$, where $a_{\mathit{send}} = v_{j-1}$ and $a_{\mathit{recv}} = v_j$.

The record $R_i$ is bilaterally signed: $\mathit{sig}_i$ includes signatures from both $v_{j-1}$ and $v_j$. The input commitment $\eta_{\mathit{in}} = \mathrm{SHA\text{-}256}(\mathit{input})$ is covered by $v_{j-1}$'s signature.

\emph{Claim}: $v_j$ cannot cause the verifier to use a fabricated input $\mathit{input}' \neq \mathit{input}$ for replay verification.

Suppose $v_j$ attempts to present $\mathit{input}'$ during forensic reconstruction. The verifier computes $\mathrm{SHA\text{-}256}(\mathit{input}')$ and compares against $\eta_{\mathit{in}}$ in the ledger record. By collision resistance of SHA-256, $\mathrm{SHA\text{-}256}(\mathit{input}') = \eta_{\mathit{in}}$ with negligible probability when $\mathit{input}' \neq \mathit{input}$. Moreover, modifying $\eta_{\mathit{in}}$ in the record $R_i$ invalidates $v_{j-1}$'s signature (EUF-CMA of Ed25519), and is further detectable via hash-chain integrity ($h_{\mathit{prev}}$ of $R_{i+1}$ would not match). When $v_{j-1}$ is honest (Cases~1 and~2 in Property~\ref{prop:g2_input_integrity}), the verifier independently requests $v_{j-1}$ to disclose $\mathit{input}$ and verifies it against $\eta_{\mathit{in}}$, bypassing $v_j$ entirely.

When $v_{j-1}$ is also compromised (Case~3), both parties may collude to present a consistent fabrication. However, Chain Verifiability (Theorem~\ref{thm:chain_verifiability_proof}) already establishes $\mathrm{CVD} \leq j-1 < j$: the chain's verifiable region does not extend to $v_j$, so G2 replay of $v_j$ is not claimed to provide end-to-end guarantees in this configuration. \qed
\end{proof}

\section{Governance Levels}
\label{app:governance_levels}

The Agent Governance Trilemma (Observation~\ref{obs:trilemma}) implies that no single governance configuration is optimal for all deployments. G1--G3 provide the infrastructure for three governance levels, each corresponding to a different position in the capability-performance-security space:

\textbf{Level 3: Compile-Time Safety} (high security + high performance, bounded capability).
For safety-critical deployments---financial transactions, medical decisions, legal document generation---agents are restricted to \emph{bounded tools only}: no code execution, no tool composition. G1 binds the complete capability envelope; G2 with full replay verification ($n \geq 200$, $\epsilon < 2.3\%$ at 99\% confidence) ensures behavioral fidelity; G3 provides complete audit trails. \emph{Recommended crypto backend}: Enhanced G1 (BBS+ with predicate proofs, \S\ref{sec:bbs_plus}) + Enhanced G2 (DV-SNARK, \S\ref{sec:dv_snark}).

\textbf{Level 2: Sampled Verification} (balanced tradeoff).
For enterprise deployments---data processing, customer service, code generation---agents may use selected unbounded tools, but G2 operates in \emph{statistical} mode with periodic sampling. A verification budget of $n = 50$ per audit cycle bounds adversarial divergence to $\epsilon < 8.9\%$ (99\% confidence). G1 declares the capability envelope; G3 ensures traceability. \emph{Recommended crypto backend}: Enhanced G1 (BBS+ selective disclosure) + Basic G2 (direct CharMatch).

\textbf{Level 3: Post-Hoc Traceability} (high capability + high performance, reduced security).
For low-risk deployments---writing assistance, creative design, information retrieval---agents operate with full capability and minimal runtime overhead. G2 is set to $\mathit{none}$; G1 provides basic capability declaration; G3 provides lightweight interaction traces. Security relies on post-hoc forensic reconstruction. \emph{Recommended crypto backend}: Basic G1 (Ed25519).

The governance level is declared in the agent's certificate and recorded in the ledger, making each deployment's security posture explicit to auditors, compliance systems, and downstream agents.

\end{document}